\documentclass[10pt,aps,prl,twocolumn,letterpaper,nobibnotes,
superscriptaddress,longbibliography,nofootinbib]{revtex4-2}

\usepackage{amsmath,amssymb,mathrsfs,mathtools,bm}
\usepackage{graphicx}
\usepackage{natbib}
\usepackage[dvipsnames]{xcolor}
\usepackage[unicode=true,bookmarks=true,breaklinks=true,
pdfborder={0 0 0},colorlinks=false,hyperfootnotes=false]{hyperref}
\hypersetup{linkcolor=magenta,urlcolor=blue,citecolor=blue,
pdfstartview={FitH}}

\newtheorem{theorem}{Theorem}
\newtheorem{lemma}[theorem]{Lemma}

\newtheorem{smcorollary}[theorem]{Corollary}
\newtheorem{proposition}[theorem]{Proposition}
\newcommand{\Tr}{\operatorname{Tr}}
\newcommand{\norm}[1]{\left\lVert#1\right\rVert}
\newcommand{\ii}{\mathrm i}
\newcommand{\im}{\mathrm i}
\newcommand{\e}{\mathrm e}

\begin{document}

\title{Theory of the Edge-of-Chaos Advantage in Quantum Reservoir Computing}

\author{Shuaifan Cao}
\affiliation{State Key Laboratory of Surface Physics, Institute of Nanoelectronics and Quantum Computing, and Department of Physics, Fudan University, Shanghai 200433, China}
\affiliation{Shanghai Qi Zhi Institute, AI Tower, Xuhui District, Shanghai 200232, China}

\author{Wei Xia}
\affiliation{Department of Physics and The Hong Kong Institute of Quantum Information Science and Technology, The Chinese University of Hong Kong, Shatin, New Territories, Hong Kong, China}

\author{Xiaopeng Li}
\email{xiaopeng\underline{ }li@fudan.edu.cn}
\affiliation{State Key Laboratory of Surface Physics, Institute of Nanoelectronics and Quantum Computing, and Department of Physics, Fudan University, Shanghai 200433, China}
\affiliation{Shanghai Qi Zhi Institute, AI Tower, Xuhui District, Shanghai 200232, China}
\affiliation{Shanghai Research Center for Quantum Sciences, Shanghai 201315, China}
\affiliation{Hefei National Laboratory, Hefei 230088, China}

\begin{abstract}
Numerical studies suggest that quantum reservoir computing performs optimally near the
edge of chaos, a behavior commonly attributed to a balance between memory retention and
nonlinear information processing. Here we develop an explicit quantum reservoir model
and rigorously analyze its chaotic regime and integrable point. In the chaotic regime,
we show that Gaussian-unitary-ensemble statistics lead to rapid memory loss: learning
performance decays exponentially with the storage time required by a temporal task. This
decay limits the ability of chaotic reservoirs to learn long-range temporal
dependencies. At the integrable point, described by a free-fermion model, we prove that
amplitude encoding and a linear readout of fermionic two-point correlations cannot learn
a representative class of nonlinear tasks coupling present and past inputs. Introducing
interactions alleviates this obstruction. By identifying distinct limitations on memory
retention in the chaotic regime and nonlinear processing at the integrable point, our
analysis provides a theoretical basis for the observed edge-of-chaos advantage and
guides the design of quantum reservoirs for temporal learning.
\end{abstract}

\maketitle

\paragraph*{Introduction.---}
Quantum machine learning (QML) has emerged as a promising frontier at the
intersection of quantum information science and artificial
intelligence~\cite{biamonte2017qml,schuld2015qml,dunjko2017qml,cerezo2022qml}.
By exploiting quantum phenomena such as superposition, entanglement, and
high-dimensional Hilbert spaces~\cite{schuld2019qml,havlicek2019qml,cong2019qcnn,jerbi2023qml}, QML provides powerful approaches
to complex data processing and pattern recognition.

Quantum reservoir computing (QRC), a quantum extension of classical reservoir
computing~\cite{jaeger2001esn,maass2002lsm,jaeger2004esn,lukosevicius2009rc,dambre2012ipc,tanaka2019rc}, is a QML framework for temporal information
processing.  It encodes input histories in a fixed quantum many-body system
and trains only a classical linear readout~\cite{fujii2017qrc,nakajima2019qrc,kutvonen2020qrc,nokkala2021gaussian,fujii2020qrc,mujal2021qrc,pena2023quantum},
thereby avoiding the barren-plateau problem associated with variational
quantum-circuit training~\cite{mcClean2018bp,cerezo2021bp,larocca2025bp}.
Its potential for nonlinear processing, memory, and temporal prediction has
been demonstrated theoretically and experimentally across diverse quantum
platforms~\cite{ghosh2019qrp,chen2020temporal,govia2021oscillator,bravo2022quantum,suzuki2022natural,martinezpena2023ipc,llodra2023particle,mujal2023measurements,domingo2023taking,kubota2023temporal,sannia2024dissipation,palacios2024coherence,cindrak2024memory,kobayashi2024feedback,senanian2024microwave,hou2026high,paparelle2026optical}.

In classical reservoir computing, reservoir performance on demanding temporal
tasks is often poor in both strongly ordered and strongly chaotic regimes, but
is enhanced near the transition between them---the ``edge of chaos''
(EOC)~\cite{langton1990eoc,bertschinger2004eoc,legenstein2007eoc,boedecker2012eoc,natschlaeger2005eoc,schrauwen2008eoc,carroll2020eoc}.  Numerical studies of QRC have found the same qualitative
pattern near localization--thermalization transitions and boundaries marking
the onset of many-body quantum
chaos~\cite{Pena2021q_eoc,xia2022q_eoc,kobayashi2026q_eoc}.  Performance
degrades toward either dynamical extreme and peaks in the crossover for tasks
requiring both memory and nonlinear processing.

This behavior is commonly attributed to a memory--expressivity tradeoff.
Regular dynamics can preserve past inputs but may fail to make the required
nonlinear temporal features accessible to the readout, whereas strongly chaotic
dynamics can generate rich transformations but rapidly erase the relevant
memory.  However, this understanding remains empirical and lacks a rigorous
theoretical foundation.

In this Letter, we analyze distinct limitations on memory retention and
nonlinear processing in an explicit model of QRC.  We first study the chaotic
regime using Hamiltonians drawn from the Gaussian unitary ensemble
(GUE)~\cite{cotler2017chaos_designs}.  For large reservoirs and sufficiently
long evolution intervals between inputs, an effective decorrelation condition
yields an exponentially decaying upper bound on ensemble-averaged memory
and, for a class of temporal targets, on learning performance as the input
delay increases.  This physical condition is verified numerically.
We expect this memory limitation to be lifted near the edge of chaos,
where the reservoir dynamics depart from the GUE description.
At the opposite limit, we study an integrable free-fermion reservoir.
Under amplitude encoding, a linear readout of fermionic two-point correlations
cannot learn a class of nonlinear tasks coupling present and past inputs.
Perturbation theory shows that weak quartic interactions can lift this
obstruction.
Guided by these results, we construct a family of tasks that probes both
limitations.  Numerical tests across an integrable-to-chaotic crossover
show an approximately exponential decrease of performance with delay in the
chaotic regime, vanishing performance at the free endpoint, and an intermediate
performance peak.  These results provide a theoretical basis for the observed edge-of-chaos
advantage and guide quantum reservoir design.

\paragraph*{Quantum reservoir computing.---}
A quantum reservoir computer processes a temporal input sequence by
alternating between state resets and unitary dynamics.
Let $\rho_{k-1}$ be the reservoir state after step $k-1$.
At step $k$, a single-qubit subsystem $A$ is reset to encode the scalar
input $u_k \in [0,1]$ via amplitude encoding:
\begin{equation}
 \rho_{\rm in}(u_k) = |\psi_k\rangle\langle\psi_k|, \quad
 |\psi_k\rangle = \sqrt{1-u_k}|0\rangle + \sqrt{u_k}|1\rangle.
 \label{eq:v2-input-encoding}
\end{equation}
The remaining reservoir $\bar{A}$ retains memory from previous steps. 
Evolving the reservoir under a fixed Hamiltonian $H$ for duration $t$
yields the updated state
\begin{equation}
 \rho_k = U\!\left[\rho_{\rm in}(u_k)\otimes \Tr_A(\rho_{k-1})\right]U^\dagger,
 \quad U = e^{-\ii Ht}.
 \label{eq:v2-reset}
\end{equation}
Here $\Tr_A$ denotes the partial trace over $A$, and
$\Tr_A(\rho_{k-1})$ is the memory state retained from the preceding steps.
A linear readout is trained to estimate the target $y_k$ as
$\widehat y_k = b + \bm w^{\mathsf T}\bm x_k$, where features
$\bm x_k := (x_{k,1},\ldots,x_{k,m})^{\mathsf T}$ are obtained from
orthonormal observables $\{O_i\}_{i=1}^m$ via $x_{k,i} = \Tr(O_i\rho_k)$
with $\Tr(O_iO_j)=\delta_{ij}$.

Under the echo-state property (ESP)~\cite{pena2023quantum}, we consider 
independent and identically distributed (i.i.d.) inputs, which 
define a unique stationary process for the reservoir state
sequence $\{\rho_k\}$.
For time-translation-invariant targets, the joint process $(\bm x_k, y_k)$
is also stationary.  With $U$ fixed, $\mathbb E_{\rm in}$ denotes expectation
over the stationary input process.
We define the feature covariance matrix $\Sigma := \operatorname{Cov}(\bm x_k,\bm x_k)$,
feature--target covariance $\bm c := \operatorname{Cov}(\bm x_k,y_k)$, and target
variance $\sigma_y^2 := \operatorname{Var}(y_k) > 0$.
Using ridge regression~\cite{sup_hoerl1970ridge} with regularization parameter $\lambda>0$, we define
$K_\lambda := (\Sigma+\lambda I)^{-1}$, where $I$ is the identity matrix in
feature space.  The population-level performance is then given by~\cite{supplemental}
\begin{equation}
 P_\lambda := 1-\mathrm{NMSE}
 = \frac{\bm c^{\mathsf T}(K_\lambda+\lambda K_\lambda^2)\bm c}{\sigma_y^2},
 \label{eq:v2-gain}
\end{equation}
where $\mathrm{NMSE} := \mathbb E_{\rm in}[(y_k-\widehat y_k)^2]/\sigma_y^2$
is the normalized mean-squared error.

\begin{figure*}[t]
\centering
\includegraphics[width=0.98\textwidth]{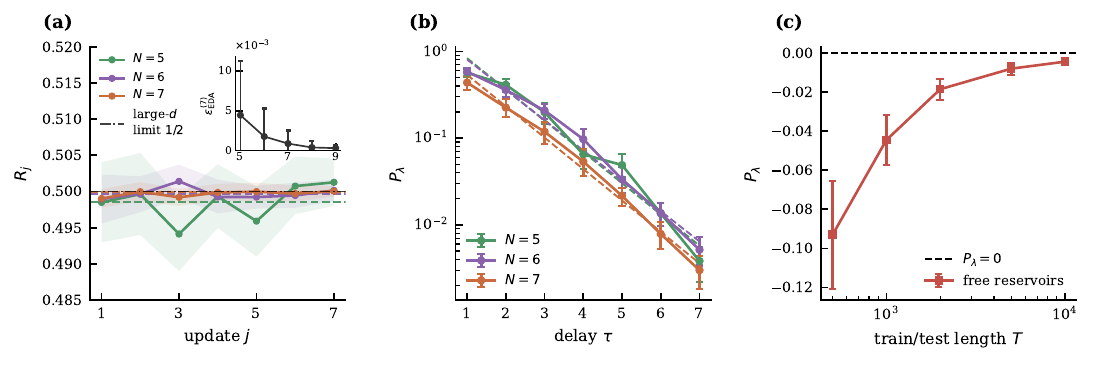}
\caption{\textbf{Numerical tests of the theoretical predictions.}
(a) \textit{Finite-window test of the EDA.}  For two runs that differ in a
single input, solid curves show direct finite-sample estimates of
$R_j=\mathbb E_{\rm joint}\|\Delta_j\|_2^2/
\mathbb E_{\rm joint}\|\Delta_{j-1}\|_2^2$ for GUE reservoirs
of $N=5,6,7$ qubits.  Colored dashed lines show
$\bar q_{t,N}$, the estimate of $\bar q_t$ for a reservoir of
$N$ qubits, obtained from the same Hamiltonian samples.  The black
dash-dotted line marks the large-$d$ limit $1/2$.  The inset shows
$\varepsilon_{{\rm EDA},N}^{(7)}
=\max_{1\leq j\leq7}|R_{j,N}-\bar q_{t,N}|$
for $N=5,\ldots,9$.
(b) \textit{Memory decay.}  Estimated performance
$P_\lambda$ for the delayed-memory task
$y_k=u_{k-\tau}$.  Dashed curves are exponential fits over
$1\leq\tau\leq7$.
(c) \textit{Free-reservoir obstruction.}  For five-qubit free reservoirs and
$y_k=(u_k-1/2)(u_{k-\tau}-1/2)$ at $\tau=1$, the estimated performance
$P_\lambda$ approaches the exact no-go value zero as the equal training
and test lengths increase.
Finite-sample estimates can be negative.
Error bars show one standard deviation across Hamiltonian realizations in
(c).  Bands and error bars in (a,b) show $95\%$ bootstrap confidence
intervals.  }
\label{fig:v2-endpoints}
\end{figure*}

\paragraph*{Chaotic regime: rapid exponential forgetting.---}
To measure how much information about a past input remains in the reservoir,
we compare two runs with the same fixed $U$ whose input sequences differ only
at step $k-\tau$.  At that step, the first run receives $u_{k-\tau}$, while
the second receives a new random number $u'_{k-\tau}$.  We draw $u'_{k-\tau}$ from the input
distribution, independently of $U$ and all inputs in the first run.
Denote their states at step $k$ by $\rho_{k,U}$ and
$\rho_{k,U}^{(k-\tau)}$, respectively.  With $U$ fixed,
$\mathbb E_{\mathrm{in}'}$ averages over the stationary input history and
$u'_{k-\tau}$.  We define the surviving influence of the information $\tau$ steps ago as
\begin{equation}
 \mathcal I_\tau(U)=\frac12\mathbb E_{\mathrm{in}'}
 \norm{\rho_{k,U}-\rho_{k,U}^{(k-\tau)}}_2^2 .
 \label{eq:v2-influence}
\end{equation}
Here $\|B\|_2^2:=\Tr(B^\dagger B)$.
For Hamiltonians $H$ drawn from an ensemble, let $\nu_t$ denote the
distribution of $U=\exp(-\ii Ht)$ at fixed $t$.
We write the ensemble-averaged influence as $\overline{\mathcal I}_\tau
:=\mathbb E_{U\sim\nu_t}\mathcal I_\tau(U)$.

To follow this influence from step $k-\tau$ through subsequent cycles of reset and evolution,
let $\Delta_j:=\Tr_A[
\rho_{k-\tau+j,U}-\rho_{k-\tau+j,U}^{(k-\tau)}]$ be the memory-state
difference of two runs at step $k-\tau+j$, for $j=0,\ldots,\tau-1$.
Thus $\Delta_0$ is the difference just after
processing $u_{k-\tau}$ and $u'_{k-\tau}$, and each increase of $j$
corresponds to one cycle with the same subsequent input in both runs.
These differences are traceless Hermitian operators on $\bar A$.
Before the final cycle, the memory difference is $\Delta_{\tau-1}$ and the influence is therefore
$\mathcal I_\tau(U)=\tfrac12\mathbb E_{\mathrm{in}'}\|\Delta_{\tau-1}\|_2^2$.

{
The evolution of the memory difference of the two runs is described by a linear map $\Phi_{u,U}(\Delta):=\Tr_A[U(\rho_{\rm in}(u)\otimes\Delta)U^\dagger]$. }
For a fixed $U$ and a nonzero memory difference $\Delta$, define the 
retention factor per step as
\begin{equation}
 q_U(\Delta):=
 \frac{\mathbb E_u\|\Phi_{u,U}(\Delta)\|_2^2}{\|\Delta\|_2^2}.
 \label{eq:v2-one-step-q}
\end{equation}
Here $\mathbb E_u$ averages over $u$, with $U$ and $\Delta$ fixed.
{
Thus $q_U$ measures the mean
ratio of the squared differences following one step. }
Writing
$\mathbb E_{\rm joint}:=\mathbb E_{U\sim\nu_t}\mathbb E_{\mathrm{in}'}$
for expectation over unitary draws, input histories, and $u'_{k-\tau}$,
we have the exact relation
$\mathbb E_{\rm joint}\|\Delta_{j+1}\|_2^2
=\mathbb E_{\rm joint}[\|\Delta_j\|_2^2q_U(\Delta_j)]$.
This relation shows how the mean squared difference evolves through one cycle.
We further define $\bar q_t:=\mathbb E_{U\sim\nu_t}[q_U]$, holding its argument
fixed while averaging the function over $U$.  

{
Because the same $U$ generates $\Delta_j$, they are generally correlated. }
However, for the chaotic reservoirs of interest here, unitary evolution spreads the
influence of past inputs over many degrees of freedom, while each reset
discards the old state of $A$ and removes its correlations with $\bar A$.
This repeated spreading and loss suggests that correlations between $U$ and $\Delta_j$ 
may have little effect in a large reservoir.  
{
The value $\bar q_t(\Delta_j)$ gives the mean retention assuming the correlation between $U$ and $\Delta_j$ is negligible. 
Quantitatively, this assumption relies on 
\begin{equation}
 \frac{\left|
 \mathbb E_{\rm joint}[\|\Delta_j\|_2^2q_U(\Delta_j)]
 -\mathbb E_{\rm joint}[\|\Delta_j\|_2^2\bar q_t(\Delta_j)]
 \right|}
 {\mathbb E_{\rm joint}\|\Delta_j\|_2^2}
 \leq\varepsilon_{\rm EDA},
 \label{eq:v2-eda}
\end{equation}
for every step $j=0,\ldots,\tau-2$.  We refer to this condition as the
effective decorrelation approximation (EDA). }
Here $\varepsilon_{\rm EDA}\geq0$ tends to zero as the memory dimension
$d_{\bar A}:=\dim\mathcal H_{\bar A}$ tends to infinity, i.e., $\varepsilon_{\rm EDA}=o(1)$.

To estimate memory retention under the EDA, we evaluate $\bar q_t(\Delta)$ for
Hamiltonians drawn from the normalized GUE~\cite{cotler2017chaos_designs}, a standard random-matrix model
of fully chaotic dynamics.  GUE conjugation invariance implies that $\bar q_t(\Delta)$
is independent of the nonzero memory difference $\Delta$.  We denote this common
value by $\bar q_t$, which depends on the evolution time $t$ and the reservoir size.
We consider evolution times $t\geq\sqrt d$, where $d=2^N$
is the reservoir Hilbert-space dimension.
Calculations using the GUE spectral form factors give
$\bar q_t=1/2+o(1)$.
Under the EDA, the factor
$\alpha_t:=\bar q_t+\varepsilon_{\rm EDA}$ bounds the ratio of the mean
squared memory differences following one step.  Thus
$\alpha_t=1/2+o(1)<1$.
The derivation is given in Supplemental Material~\cite{supplemental}.

In the numerical results, we use the same notation for quantities
and their finite-sample estimates.
We numerically validate the EDA for normalized GUE reservoirs at
evolution times $t\geq\sqrt d$ in Fig.~\ref{fig:v2-endpoints}(a).
We compare the retention ratio
$R_j:=\mathbb E_{\rm joint}\|\Delta_j\|_2^2/
\mathbb E_{\rm joint}\|\Delta_{j-1}\|_2^2$ with the mean retention factor
$\bar q_t$.  For GUE reservoirs, the EDA reduces to
$|R_j-\bar q_t|\leq\varepsilon_{\rm EDA}$.
Their estimated discrepancy decreases as the
reservoir size increases from $N=5$ to $9$.  The estimated values of $\alpha_t$ remain below $0.504$.
In extended tests with $N=6,7,8$, $R_j$ remains close to $1/2$ through 100 updates
(see Supplemental Material~\cite{supplemental}).

To connect this memory loss to task performance, we consider targets for
which information about the past input $u_{k-\tau}$ is essential for
prediction.  We call $y_k$ a lag-$\tau$ innovation target if
$\mathbb E_{\rm in}[y_k\mid u_j:j\leq k,\ j\neq k-\tau]
=\mathbb E_{\rm in}y_k$, which means that the best prediction based on 
the other inputs alone is merely the target mean.

\begin{theorem}[GUE forgetting under the EDA]
\label{thm:v2-gue-forgetting}
Let $\tau\geq1$ be an integer and $y_k$ a lag-$\tau$ innovation target.
For GUE reservoirs with evolution time $t\geq\sqrt d$, if the EDA holds for
$j=0,\ldots,\tau-2$ (an empty condition for $\tau=1$), then
\begin{equation}
 \mathbb E_{U\sim\nu_t}P_\lambda(U)
 \leq\frac{2\overline{\mathcal I}_1}{\lambda}
 \left[\frac{1}{2}+o(1)\right]^{\tau-1},
 \label{eq:v2-gue-forgetting}
\end{equation}
where $o(1)$ tends to zero as $d_{\bar A}\to\infty$.
The performance bound decays exponentially with lag over the range covered
by the EDA.
\end{theorem}

For an innovation target, positive $P_\lambda$ requires the reservoir to
retain information about $u_{k-\tau}$.  The surviving influence
$\mathcal I_\tau(U)$ measures how much replacing $u_{k-\tau}$ changes the
reservoir state and provides a bound on $P_\lambda$.  Applying the EDA to
each of the $\tau-1$ further cycles bounds the mean remaining influence by
$\overline{\mathcal I}_1\alpha_t^{\tau-1}$.  Substituting the GUE value of
$\alpha_t$ gives the stated performance bound.  Details 
are given in Supplemental Material~\cite{supplemental}.

Figure~\ref{fig:v2-endpoints}(b) shows that performance on the
delayed-memory target $y_k=u_{k-\tau}$ decreases approximately exponentially
over $1\leq\tau\leq7$, consistent with the exponential forgetting bound in
Theorem~\ref{thm:v2-gue-forgetting}.  Numerical details are given in
Supplemental Material~\cite{supplemental}.

The GUE result establishes rapid exponential forgetting under the EDA,
with the bound on ensemble-averaged memory decreasing by a factor
approaching $1/2$ per cycle.  This severe memory loss imposes a strong
limitation on prediction at long delays.  
We expect this limitation to be
lifted as the reservoir leaves the chaotic regime and its dynamics are no
longer well described by the GUE.
We next consider the opposite, integrable limit, taking a free-fermion
reservoir as a representative example.

\paragraph*{Free endpoint: expressivity obstruction.---}
{
At the free endpoint, the reservoir evolves under a quadratic fermionic
Hamiltonian and is read out through quadratic observables.  }
We use the Jordan--Wigner transformation to represent the
$N$-qubit reservoir in terms of $2N$ Majorana operators
$\{m_p\}_{p=1}^{2N}$.  We choose the reset qubit $A$
as the first qubit in the Jordan--Wigner ordering.  With this choice,
\begin{equation}
 m_1=X_A\otimes I_{\bar{A}},\quad
 m_2=Y_A\otimes I_{\bar{A}},\quad
 \ii m_1m_2=-Z_A\otimes I_{\bar{A}}.
 \label{eq:v2-reset-majoranas}
\end{equation}
Here $X_A$, $Y_A$, and $Z_A$ are the Pauli operators on $A$, and
$I_{\bar{A}}$ is the identity on $\bar A$.  

Under the amplitude encoding in Eq.~\eqref{eq:v2-input-encoding}, the reset
state can be written as
\begin{equation}
 \rho_{\rm in}(u_k)
 =\frac12\left[I_A+\phi(u_k)X_A+(1-2u_k)Z_A\right],
 \label{eq:v2-free-reset-state}
\end{equation}
where $\phi(u):=2\sqrt{u(1-u)}$.  Therefore, the dependence of any feature on the current
input can involve only $\phi(u_k)$ and $z_k:=u_k-1/2$. 

We now consider whether the reservoir can learn certain cross-time products
involving current and past inputs.  For an input
distribution symmetric about $1/2$, $z_k$ has zero mean and is orthogonal
to $\phi(u_k)$.  We study targets of the form
\begin{equation}
 y_k=z_kh_{k-1},
 \label{eq:v2-free-target-class}
\end{equation}
where $h_{k-1}:=h(\ldots,u_{k-2},u_{k-1})$ is a nonzero, centered and 
square-integrable function of earlier inputs.  
A function is called centered
if it has zero expectation, i.e.,
$\mathbb E_{\rm in}h_{k-1}=0$.  
The following theorem shows that a free reservoir cannot learn any
nonzero target of this type.

\begin{theorem}[Cross-time product no-go]
\label{thm:v2-main-free-nogo}
Under the amplitude encoding in Eq.~\eqref{eq:v2-input-encoding}, let the
inputs be i.i.d. with a distribution symmetric about $1/2$.  For any
quadratic Hamiltonian and any collection of quadratic observables, every
nonzero target in Eq.~\eqref{eq:v2-free-target-class} has performance
\begin{equation}
 P_\lambda=0 .
 \label{eq:v2-nogo}
\end{equation}
\end{theorem}

The proof follows from quadratic closure.  Any quadratic observable remains
quadratic under Heisenberg evolution generated by a quadratic Hamiltonian.
The term linear in the current input $z_k$ enters the
reservoir only through $Z_A$.  Within the
quadratic span, the only term containing $Z_A$ is
$\ii m_1m_2=-Z_A\otimes I_{\bar A}$, which acts trivially on $\bar A$.
Therefore, the coefficient of $z_k$ cannot depend on earlier inputs
stored in $\bar A$, and the reservoir cannot learn the cross-time dependence
in Eq.~\eqref{eq:v2-free-target-class}.
{
Details are given in Supplemental Material~\cite{supplemental}. }

Although the theorem is stated for centered targets, it also limits 
uncentered targets containing the same inaccessible component.  
For example, the uncentered target
$y_k=u_ku_{k-\tau}$ with an integer delay $\tau\geq1$ obeys the nontrivial 
upper bound $P_\lambda\leq6/7$ under uniform inputs~\cite{supplemental}.

Under the QRC rewinding
protocol~\cite{mujal2023measurements}, perturbation theory shows that weak
quartic interactions can lift the exact obstruction in
Theorem~\ref{thm:v2-main-free-nogo}. 
For inputs with a distribution symmetric about $1/2$, consider the centered delayed-product
target $y_k=z_kz_{k-\tau}$ with $N\geq3$ and $1\leq\tau\leq N-1$.
Let $H_2$ be a quadratic Hamiltonian and $H_4$ a pure-quartic Majorana
Hamiltonian.  We interpolate between them as
\begin{equation}
 H_\theta=\cos\theta\,H_2+\sin\theta\,H_4,
 \label{eq:v2-interpolation}
\end{equation}
so that $\theta=0$ is the free endpoint, and small $\theta$ introduces a weak
quartic interaction.  For the evolution
$U_\theta=\exp(-\ii H_\theta t)$, let $\bm c(\theta)$ be the corresponding
feature--target covariance under this protocol.  At $\theta=0$, the
no-go argument gives $\bm c(0)=0$.  
{
There exist choices of $H_2$, $H_4$, the
evolution time $t$, and a collection of quadratic observables for which
$\dot{\bm c}(0)\neq0$, where the dot denotes differentiation with respect to
$\theta$ (see Supplemental Material~\cite{supplemental}). }
Consequently, for every $\lambda>0$, 
\begin{equation}
 P_\lambda(\theta)=\kappa_\lambda\theta^2+O(\theta^3),
 \qquad \kappa_\lambda>0.
 \label{eq:v2-perturbative-gain}
\end{equation}
Thus, a sufficiently small nonzero quartic interaction gives $P_\lambda>0$
under the rewinding protocol. 

Figure~\ref{fig:v2-endpoints}(c) numerically illustrates the free-endpoint
obstruction.  For quadratic Hamiltonians and
$y_k=z_kz_{k-1}$, the estimated performance $P_\lambda$ approaches zero as
the training and test lengths increase.
Since these reservoirs can still predict the linear target $y_k=u_{k-\tau}$
(see Supplemental Material~\cite{supplemental}), this failure reflects the
nonlinearity obstruction.

Near the edge of chaos, both the rapid memory loss characteristic of the
chaotic regime and the nonlinear-processing obstruction at the free endpoint
can be alleviated.  We expect the resulting balance between memory retention
and nonlinear processing to underlie the optimal performance observed in
practical quantum reservoir computing tasks.

\paragraph*{Theory-guided tasks for probing the edge of chaos.---}
Guided by the two theorems, we construct a family of nonlinear tasks that 
probe both memory loss in the chaotic regime and the inability to learn certain 
nonlinear cross-time products at the free endpoint.   For uniform inputs, we define
\begin{equation}
 y_k^{(S)}=(u_k-1/2)\prod_{s\in S}(u_{k-s}-1/2),
 \label{eq:v2-probing-task-family}
\end{equation}
where $S$ is any nonempty finite set of positive integer delays, specifying
which past inputs enter the product.  Each target belongs to the
free-reservoir no-go class and is a lag-$\tau$ innovation target for every
$\tau\in S$.
For the numerical tests, we choose the simplest case, $S=\{\tau\}$ with
$\tau\geq1$, giving $y_k=(u_k-1/2)(u_{k-\tau}-1/2)$.
We compute the stationary performance $P_\lambda$ using quadratic observables along
the quadratic--quartic Sachdev--Ye--Kitaev (SYK) interpolation in
Eq.~\eqref{eq:v2-interpolation}, where $g:=\tan\theta$ measures the strength
of the quartic term relative to the quadratic term.  This defines an
integrable-to-chaotic crossover from free
fermions at $g=0$ to the pure-quartic $\mathrm{SYK}_4$ model at $g=\infty$,
a standard model of many-body quantum
chaos~\cite{maldacena2016remarks,garciagarcia2018chaotic,haque2019eigenstate,kobayashi2026q_eoc}.
The Hamiltonian normalization and numerical methods are given in
Supplemental Material~\cite{supplemental}.

\begin{figure}[t]
\centering
\includegraphics[width=0.90\columnwidth]{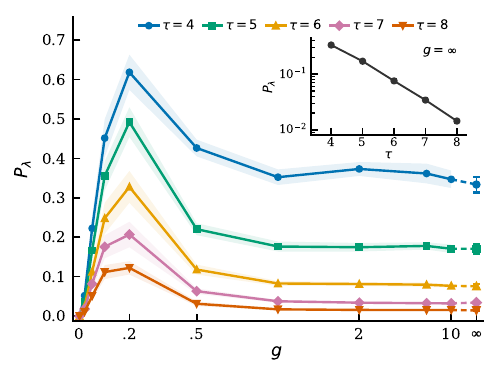}
\caption{\textbf{Theory-guided tasks reveal an intermediate performance peak.}
Stationary $P_\lambda$ for
$y_k=(u_k-1/2)(u_{k-\tau}-1/2)$ with $4\leq\tau\leq8$.
The readout uses the complete set of normalized quadratic observables.
The horizontal axis is scaled as $\arctan g$.
Dashed terminal segments connect $g=10$ to the separately evaluated
limit $g=\infty$ in the chaotic regime.  The inset replots the $g=\infty$
values on a logarithmic scale.  Values are sample means over the
same 30 independently drawn $N=6$ mixed
$\mathrm{SYK}_2+\mathrm{SYK}_4$ Hamiltonian pairs at $t=2$ and
$\lambda\simeq1.56\times10^{-6}$.  Shading and error bars show
$95\%$ confidence intervals for the ensemble mean.}
\label{fig:v2-crossover}
\end{figure}

Figure~\ref{fig:v2-crossover} shows the two limitations.
At the free endpoint $g=0$, $P_\lambda=0$ for every displayed lag,
in agreement with Theorem~\ref{thm:v2-main-free-nogo}.
In the chaotic regime at $g\gtrsim1$, performance decreases approximately
exponentially with the delay $\tau$ across the sampled couplings; the inset
illustrates this behavior at $g=\infty$.  This trend supports the memory-loss
mechanism identified by the GUE analysis.

Away from the free endpoint, performance improves as quartic interactions are introduced
and reaches its largest sampled value at $g=0.2$ for every displayed lag.
Performance then declines toward the chaotic regime.  The intermediate
region outperforms both the free endpoint and the sampled chaotic regime,
supporting the expected advantage of
combining useful memory with the nonlinear response required by the task.

\paragraph*{Discussion.---}
We have established distinct limitations on memory retention in chaotic quantum reservoirs and nonlinear processing at an integrable point. 
Their simultaneous lifting near the edge of chaos is expected to be the key to optimal performance in practical quantum reservoir computing tasks, thus providing a theoretical basis for the observed edge-of-chaos advantage.
Although our rigorous results concern specific task families, the underlying obstructions are
relevant more broadly to temporal learning tasks that require both persistent memory and
nonlinear coupling between inputs at different times. For example, the input-product
terms in NARMA tasks share the nonlinear cross-time structure of the task family studied
here~\cite{kobayashi2026q_eoc}. 
Our framework thus identifies mechanisms that can account for the edge-of-chaos advantage observed in such tasks. 
These results suggest a design
principle for quantum reservoirs: introduce sufficient interactions to enable nonlinear
temporal processing while avoiding the rapid memory loss associated with strongly
chaotic dynamics.


\paragraph*{Acknowledgments.---}
We thank Xingze Qiu for helpful discussions. This work was supported by the National
Basic Research Program of China (Grant No.~2021YFA1400900), and the Innovation Program
for Quantum Science and Technology of China (Grant No.~2024ZD0300100), and Shanghai
Municipal Science and Technology (Grant Nos.~26LZ0500100, 25TQ003, and 24DP2600100).

\nocite{oliviero2021isospectral_twirling,soshnikov2000determinantal,sup_almohy2009frechet,sup_bhatia1997matrix,sup_bravyi2005fermionic,sup_collins2006integration,sup_efron1979bootstrap,sup_efronstein1981jackknife,sup_garoni2005density,sup_mele2024swap,sup_mityagin2020zero}
\bibliographystyle{apsrev4-1}
\bibliography{references}

@article{maldacena2016remarks,
  title = {Remarks on the {Sachdev-Ye-Kitaev} model},
  author = {Maldacena, Juan and Stanford, Douglas},
  journal = {Physical Review D},
  volume = {94},
  number = {10},
  pages = {106002},
  year = {2016},
  doi = {10.1103/PhysRevD.94.106002},
  eprint = {1604.07818},
  archivePrefix = {arXiv},
  primaryClass = {hep-th}
}

@article{garciagarcia2018chaotic,
  title = {Chaotic-Integrable Transition in the {Sachdev-Ye-Kitaev} Model},
  author = {Garc{\'i}a-Garc{\'i}a, Antonio M. and Loureiro, Bruno and Romero-Berm{\'u}dez, Aurelio and Tezuka, Masaki},
  journal = {Physical Review Letters},
  volume = {120},
  number = {24},
  pages = {241603},
  year = {2018},
  doi = {10.1103/PhysRevLett.120.241603},
  eprint = {1707.02197},
  archivePrefix = {arXiv},
  primaryClass = {hep-th}
}

@article{haque2019eigenstate,
  title = {Eigenstate Thermalization Scaling in Majorana Clusters: From Chaotic to Integrable {Sachdev-Ye-Kitaev} Models},
  author = {Haque, Masudul and McClarty, Paul A.},
  journal = {Physical Review B},
  volume = {100},
  number = {11},
  pages = {115122},
  year = {2019},
  doi = {10.1103/PhysRevB.100.115122},
  eprint = {1711.02360},
  archivePrefix = {arXiv},
  primaryClass = {cond-mat.stat-mech}
}

@article{langton1990eoc,
  author  = {Langton, Chris G.},
  title   = {Computation at the edge of chaos: Phase transitions and emergent computation},
  journal = {Physica D: Nonlinear Phenomena},
  volume  = {42},
  number  = {1--3},
  pages   = {12--37},
  year    = {1990},
  doi     = {10.1016/0167-2789(90)90064-V}
}

@inproceedings{natschlaeger2005eoc,
  author    = {Bertschinger, Nils and Natschl{\"a}ger, Thomas and Legenstein, Robert A.},
  title     = {At the edge of chaos: Real-time computations and self-organized criticality in recurrent neural networks},
  booktitle = {Advances in Neural Information Processing Systems},
  volume    = {17},
  pages     = {145--152},
  year      = {2005},
  publisher = {MIT Press}
}

@article{legenstein2007eoc,
  author  = {Legenstein, Robert and Maass, Wolfgang},
  title   = {Edge of chaos and prediction of computational performance for neural circuit models},
  journal = {Neural Networks},
  volume  = {20},
  number  = {3},
  pages   = {323--334},
  year    = {2007},
  doi     = {10.1016/j.neunet.2007.04.017}
}

@article{boedecker2012eoc,
  author  = {Boedecker, Joschka and Obst, Oliver and Lizier, Joseph T. and Mayer, N. Michael and Asada, Minoru},
  title   = {Information processing in echo state networks at the edge of chaos},
  journal = {Theory in Biosciences},
  volume  = {131},
  number  = {3},
  pages   = {205--213},
  year    = {2012},
  doi     = {10.1007/s12064-011-0146-8}
}

@article{pena2023quantum,
  title = {Quantum reservoir computing in finite dimensions},
  author = {Mart\'{\i}nez-Pe\~na, Rodrigo and Ortega, Juan-Pablo},
  journal = {Phys. Rev. E},
  volume = {107},
  issue = {3},
  pages = {035306},
  numpages = {16},
  year = {2023},
  month = {Mar},
  publisher = {American Physical Society},
  doi = {10.1103/PhysRevE.107.035306},
  url = {https://link.aps.org/doi/10.1103/PhysRevE.107.035306}
}

@ARTICLE{domingo2023taking,
  title    = "Taking advantage of noise in quantum reservoir computing",
  author   = "Domingo, L and Carlo, G and Borondo, F",
  journal  = "Scientific Reports",
  volume   =  13,
  number   =  1,
  pages    = "8790",
  month    =  may,
  year     =  2023
}

@article{bravo2022quantum,
  title = {Quantum Reservoir Computing Using Arrays of Rydberg Atoms},
  author = {Bravo, Rodrigo Araiza and Najafi, Khadijeh and Gao, Xun and Yelin, Susanne F.},
  journal = {PRX Quantum},
  volume = {3},
  issue = {3},
  pages = {030325},
  numpages = {19},
  year = {2022},
  month = {Aug},
  publisher = {American Physical Society},
  doi = {10.1103/PRXQuantum.3.030325},
  url = {https://link.aps.org/doi/10.1103/PRXQuantum.3.030325}
}

@article{kubota2023temporal,
  title = {Temporal information processing induced by quantum noise},
  author = {Kubota, Tomoyuki and Suzuki, Yudai and Kobayashi, Shumpei and Tran, Quoc Hoan and Yamamoto, Naoki and Nakajima, Kohei},
  journal = {Phys. Rev. Res.},
  volume = {5},
  issue = {2},
  pages = {023057},
  numpages = {15},
  year = {2023},
  month = {Apr},
  publisher = {American Physical Society},
  doi = {10.1103/PhysRevResearch.5.023057},
  url = {https://link.aps.org/doi/10.1103/PhysRevResearch.5.023057}
}

@article{hou2026high,
  title = {High-Accuracy Temporal Prediction via Experimental Quantum Reservoir Computing in Correlated Spins},
  author = {Hou, Yanjun and Hua, Juncheng and Wu, Ze and Xia, Wei and Chen, Yuquan and Li, Xiaopeng and Li, Zhaokai and Peng, Xinhua and Du, Jiangfeng},
  journal = {Phys. Rev. Lett.},
  volume = {136},
  issue = {12},
  pages = {120602},
  numpages = {8},
  year = {2026},
  month = {Mar},
  publisher = {American Physical Society},
  doi = {10.1103/r8ww-qw7j},
  url = {https://link.aps.org/doi/10.1103/r8ww-qw7j}
}

@ARTICLE{suzuki2022natural,
  title    = "Natural quantum reservoir computing for temporal information
              processing",
  author   = "Suzuki, Yudai and Gao, Qi and Pradel, Ken C and Yasuoka, Kenji
              and Yamamoto, Naoki",
  journal  = "Scientific Reports",
  volume   =  12,
  number   =  1,
  pages    = "1353",
  month    =  jan,
  year     =  2022,
  doi      = {10.1038/s41598-022-05061-w}
}

@article{chen2020temporal,
  title = {Temporal Information Processing on Noisy Quantum Computers},
  author = {Chen, Jiayin and Nurdin, Hendra I. and Yamamoto, Naoki},
  journal = {Phys. Rev. Appl.},
  volume = {14},
  issue = {2},
  pages = {024065},
  numpages = {22},
  year = {2020},
  month = {Aug},
  publisher = {American Physical Society},
  doi = {10.1103/PhysRevApplied.14.024065},
  url = {https://link.aps.org/doi/10.1103/PhysRevApplied.14.024065}
}

@article{mujal2023measurements,
  title={Time-series quantum reservoir computing with weak and projective measurements},
  author={Mujal, Pere and Mart{\'\i}nez-Pe{\~n}a, Rodrigo and Giorgi, Gian Luca and Soriano, Miguel C. and Zambrini, Roberta},
  journal={npj Quantum Information},
  volume={9},
  pages={16},
  year={2023},
  publisher={Nature Publishing Group},
  doi={10.1038/s41534-023-00682-z}
}

@article{sannia2024dissipation,
  title={Dissipation as a resource for {Q}uantum {R}eservoir {C}omputing},
  author={Sannia, Antonio and Mart{\'\i}nez-Pe{\~n}a, Rodrigo and Soriano, Miguel C. and Giorgi, Gian Luca and Zambrini, Roberta},
  journal={Quantum},
  volume={8},
  pages={1291},
  year={2024},
  publisher={Verein zur F{\"o}rderung des Open Access Publizierens in den Quantenwissenschaften},
  doi={10.22331/q-2024-03-20-1291}
}

@article{palacios2024coherence,
  title={Role of coherence in many-body {Q}uantum {R}eservoir {C}omputing},
  author={Palacios, Ana and Mart{\'\i}nez-Pe{\~n}a, Rodrigo and Soriano, Miguel C. and Giorgi, Gian Luca and Zambrini, Roberta},
  journal={Communications Physics},
  volume={7},
  number={1},
  pages={369},
  year={2024},
  publisher={Nature Publishing Group UK}
}

@ARTICLE{tanaka2019rc,
  title    = "Recent advances in physical reservoir computing: A review",
  author   = "Tanaka, Gouhei and Yamane, Toshiyuki and H{\'e}roux, Jean Benoit
              and Nakane, Ryosho and Kanazawa, Naoki and Takeda, Seiji and
              Numata, Hidetoshi and Nakano, Daiju and Hirose, Akira",
  journal  = "Neural Networks",
  volume   =  115,
  pages    = "100--123",
  month    =  jul,
  year     =  2019
}

@ARTICLE{dambre2012ipc,
  title    = "Information Processing Capacity of Dynamical Systems",
  author   = "Dambre, Joni and Verstraeten, David and Schrauwen, Benjamin and
              Massar, Serge",
  journal  = "Scientific Reports",
  volume   =  2,
  number   =  1,
  pages    = "514",
  month    =  jul,
  year     =  2012
}

@techreport{jaeger2001esn,
  title={The ``echo state'' approach to analysing and training recurrent neural networks-with an erratum note},
  author={Jaeger, Herbert},
  year={2001},
  institution={German National Research Center for Information Technology GMD Technical Report 148},
  address={Bonn, Germany}
}

@article{maass2002lsm,
  title={Real-time computing without stable states: A new framework for neural computation based on perturbations},
  author={Maass, Wolfgang and Natschl{\"a}ger, Thomas and Markram, Henry},
  journal={Neural computation},
  volume={14},
  number={11},
  pages={2531--2560},
  year={2002},
  publisher={MIT Press}
}

@article{jaeger2004esn,
  title={Harnessing nonlinearity: Predicting chaotic systems and saving energy in wireless communication},
  author={Jaeger, Herbert and Haas, Harald},
  journal={Science},
  volume={304},
  number={5667},
  pages={78--80},
  year={2004},
  publisher={American Association for the Advancement of Science}
}

@article{lukosevicius2009rc,
  title={Reservoir computing approaches to recurrent neural network training},
  author={Luko{\v{s}}evi{\v{c}}ius, Mantas and Jaeger, Herbert},
  journal={Computer science review},
  volume={3},
  number={3},
  pages={127--149},
  year={2009},
  publisher={Elsevier}
}

@ARTICLE{larocca2025bp,
  title    = "Barren plateaus in variational quantum computing",
  author   = "Larocca, Mart{\'\i}n and Thanasilp, Supanut and Wang, Samson and
              Sharma, Kunal and Biamonte, Jacob and Coles, Patrick J and
              Cincio, Lukasz and McClean, Jarrod R and Holmes, Zo{\"e} and
              Cerezo, M",
  journal  = "Nature Reviews Physics",
  volume   =  7,
  number   =  4,
  pages    = "174--189",
  month    =  apr,
  year     =  2025,
  doi      = {10.1038/s42254-025-00813-9}
}

@ARTICLE{biamonte2017qml,
  title    = "Quantum machine learning",
  author   = "Biamonte, Jacob and Wittek, Peter and Pancotti, Nicola and
              Rebentrost, Patrick and Wiebe, Nathan and Lloyd, Seth",
  journal  = "Nature",
  volume   =  549,
  number   =  7671,
  pages    = "195--202",
  month    =  sep,
  year     =  2017
}

@ARTICLE{schuld2015qml,
  title     = "An introduction to quantum machine learning",
  author    = "Schuld, Maria and Sinayskiy, Ilya and Petruccione, Francesco",
  journal   = "Contemporary Physics",
  publisher = "Taylor \& Francis",
  volume    =  56,
  number    =  2,
  pages     = "172--185",
  month     =  apr,
  year      =  2015
}

@article{dunjko2017qml,
  title={Machine learning \& artificial intelligence in the quantum domain: a review of recent progress},
  author={Dunjko, Vedran and Briegel, Hans J.},
  journal={Reports on Progress in Physics},
  volume={81},
  number={7},
  pages={074001},
  year={2018},
  doi={10.1088/1361-6633/aab406}
}

@ARTICLE{jerbi2023qml,
  title    = "Quantum machine learning beyond kernel methods",
  author   = "Jerbi, Sofiene and Fiderer, Lukas J and Poulsen Nautrup, Hendrik
              and K{\"u}bler, Jonas M and Briegel, Hans J and Dunjko, Vedran",
  journal  = "Nature Communications",
  volume   =  14,
  number   =  1,
  pages    = "517",
  month    =  jan,
  year     =  2023
}

@article{schuld2019qml,
  title = {Quantum Machine Learning in Feature Hilbert Spaces},
  author = {Schuld, Maria and Killoran, Nathan},
  journal = {Phys. Rev. Lett.},
  volume = {122},
  issue = {4},
  pages = {040504},
  numpages = {6},
  year = {2019},
  month = {Feb},
  publisher = {American Physical Society},
  doi = {10.1103/PhysRevLett.122.040504},
  url = {https://link.aps.org/doi/10.1103/PhysRevLett.122.040504}
}

@ARTICLE{mcClean2018bp,
  title    = "Barren plateaus in quantum neural network training landscapes",
  author   = "McClean, Jarrod R and Boixo, Sergio and Smelyanskiy, Vadim N and
              Babbush, Ryan and Neven, Hartmut",
  journal  = "Nature Communications",
  volume   =  9,
  number   =  1,
  pages    = "4812",
  month    =  nov,
  year     =  2018
}

@ARTICLE{cerezo2021bp,
  title    = "Cost function dependent barren plateaus in shallow parametrized
              quantum circuits",
  author   = "Cerezo, M and Sone, Akira and Volkoff, Tyler and Cincio, Lukasz
              and Coles, Patrick J",
  journal  = "Nature Communications",
  volume   =  12,
  number   =  1,
  pages    = "1791",
  month    =  mar,
  year     =  2021
}

@article{fujii2017qrc,
  author  = {Keisuke Fujii and Kohei Nakajima},
  title   = {Harnessing Disordered-Ensemble Quantum Dynamics for Machine Learning},
  journal = {Physical Review Applied},
  volume  = {8},
  number  = {2},
  pages   = {024030},
  year    = {2017},
  doi     = {10.1103/PhysRevApplied.8.024030}
}

@article{nakajima2019qrc,
  author  = {Kohei Nakajima and Keisuke Fujii and Makoto Negoro and Kosuke Mitarai and Masahiro Kitagawa},
  title   = {Boosting Computational Power through Spatial Multiplexing in Quantum Reservoir Computing},
  journal = {Physical Review Applied},
  volume  = {11},
  number  = {3},
  pages   = {034021},
  year    = {2019},
  doi     = {10.1103/PhysRevApplied.11.034021}
}

@article{kutvonen2020qrc,
  author  = {Antti Kutvonen and Keisuke Fujii and Takahiro Sagawa},
  title   = {Optimizing a Quantum Reservoir Computer for Time Series Prediction},
  journal = {Scientific Reports},
  volume  = {10},
  pages   = {14687},
  year    = {2020},
  doi     = {10.1038/s41598-020-71601-1}
}

@incollection{fujii2020qrc,
  author    = {Keisuke Fujii and Kohei Nakajima},
  title     = {Quantum Reservoir Computing: A Reservoir Approach toward Quantum Machine Learning on Near-Term Quantum Devices},
  booktitle = {Reservoir Computing},
  editor    = {Kohei Nakajima and Ingo Fischer},
  publisher = {Springer},
  address   = {Singapore},
  year      = {2021},
  pages     = {423--450},
  doi       = {10.1007/978-981-13-1687-6_23}
}

@article{mujal2021qrc,
  author  = {Pol Mujal and Miguel {C. Mu{\~n}oz} and Carlos Navarrete-Benlloch and Adri{\'a}n Pérez-Salinas and Diego García-Martín},
  title   = {Opportunities in Quantum Reservoir Computing and Extreme Learning Machines},
  journal = {Advanced Quantum Technologies},
  volume  = {4},
  number  = {8},
  pages   = {2100027},
  year    = {2021},
  doi     = {10.1002/qute.202100027}
}

@article{bertschinger2004eoc,
  author  = {Nils Bertschinger and Thomas Natschl{\"a}ger},
  title   = {Real-Time Computation at the Edge of Chaos in Recurrent Neural Networks},
  journal = {Neural Computation},
  volume  = {16},
  number  = {7},
  pages   = {1413--1436},
  year    = {2004},
  doi     = {10.1162/089976604323057443}
}

@inproceedings{schrauwen2008eoc,
  author    = {Benjamin Schrauwen and Lars B{\"u}sing and Robert Legenstein},
  title     = {On Computational Power and the Order-Chaos Phase Transition in Reservoir Computing},
  booktitle = {Advances in Neural Information Processing Systems 21},
  year      = {2008},
  pages     = {1425--1432}
}

@ARTICLE{carroll2020eoc,
  title    = "Do reservoir computers work best at the edge of chaos?",
  author   = "Carroll, T L",
  journal  = "Chaos: An Interdisciplinary Journal of Nonlinear Science",
  volume   =  30,
  number   =  12,
  pages    = "121109",
  month    =  dec,
  year     =  2020
}

@article{Pena2021q_eoc,
  title = {Dynamical Phase Transitions in Quantum Reservoir Computing},
  author = {Mart\'{\i}nez-Pe\~na, Rodrigo and Giorgi, Gian Luca and Nokkala, Johannes and Soriano, Miguel C. and Zambrini, Roberta},
  journal = {Phys. Rev. Lett.},
  volume = {127},
  issue = {10},
  pages = {100502},
  numpages = {7},
  year = {2021},
  month = {Aug},
  publisher = {American Physical Society},
  doi = {10.1103/PhysRevLett.127.100502},
  url = {https://link.aps.org/doi/10.1103/PhysRevLett.127.100502}
}

@article{kobayashi2026q_eoc,
  title = {Edge of Many-Body Quantum Chaos in Quantum Reservoir Computing},
  author = {Kobayashi, Kaito and Motome, Yukitoshi},
  journal = {Phys. Rev. Lett.},
  volume = {136},
  issue = {4},
  pages = {040602},
  numpages = {8},
  year = {2026},
  month = {Jan},
  publisher = {American Physical Society},
  doi = {10.1103/j2qj-vwcl},
  url = {https://link.aps.org/doi/10.1103/j2qj-vwcl}
}

@ARTICLE{xia2022q_eoc,
  title    = "The reservoir learning power across quantum many-body
              localization transition",
  author   = "Xia, Wei and Zou, Jie and Qiu, Xingze and Li, Xiaopeng",
  journal  = "Frontiers of Physics",
  volume   =  17,
  number   =  3,
  pages    = "33506",
  month    =  apr,
  year     =  2022
}

@article{cotler2017chaos_designs,
  author = {Cotler, Jordan S. and Hunter-Jones, Nicholas and Liu, Junyu and Yoshida, Beni},
  title = {Chaos, Complexity, and Random Matrices},
  journal = {Journal of High Energy Physics},
  volume = {2017},
  number = {11},
  pages = {048},
  year = {2017},
  doi = {10.1007/JHEP11(2017)048},
  eprint = {1706.05400},
  archivePrefix = {arXiv},
  primaryClass = {hep-th}
}

@article{oliviero2021isospectral_twirling,
  author = {Oliviero, Salvatore F. E. and Leone, Lorenzo and Caravelli, Francesco and Hamma, Alioscia},
  title = {Random Matrix Theory of the Isospectral Twirling},
  journal = {SciPost Physics},
  volume = {10},
  pages = {076},
  year = {2021},
  doi = {10.21468/SciPostPhys.10.3.076},
  eprint = {2012.07681},
  archivePrefix = {arXiv},
  primaryClass = {quant-ph}
}

@article{soshnikov2000determinantal,
  author = {Soshnikov, Alexander},
  title = {Determinantal Random Point Fields},
  journal = {Russian Mathematical Surveys},
  volume = {55},
  number = {5},
  pages = {923--975},
  year = {2000},
  doi = {10.1070/RM2000v055n05ABEH000321},
  eprint = {math/0002099},
  archivePrefix = {arXiv},
  primaryClass = {math.PR}
}

@article{cerezo2022qml,
  author = {Cerezo, M. and Verdon, Guillaume and Huang, Hsin-Yuan and Cincio, Lukasz and Coles, Patrick J.},
  title = {Challenges and Opportunities in Quantum Machine Learning},
  journal = {Nature Computational Science},
  volume = {2},
  number = {9},
  pages = {567--576},
  year = {2022},
  doi = {10.1038/s43588-022-00311-3}
}

@article{havlicek2019qml,
  author = {Havl{\'i}{\v c}ek, Vojt{\v e}ch and C{\'o}rcoles, Antonio D. and Temme, Kristan and Harrow, Aram W. and Kandala, Abhinav and Chow, Jerry M. and Gambetta, Jay M.},
  title = {Supervised Learning with Quantum-Enhanced Feature Spaces},
  journal = {Nature},
  volume = {567},
  number = {7747},
  pages = {209--212},
  year = {2019},
  doi = {10.1038/s41586-019-0980-2}
}

@article{cong2019qcnn,
  author = {Cong, Iris and Choi, Soonwon and Lukin, Mikhail D.},
  title = {Quantum Convolutional Neural Networks},
  journal = {Nature Physics},
  volume = {15},
  number = {12},
  pages = {1273--1278},
  year = {2019},
  doi = {10.1038/s41567-019-0648-8}
}

@article{ghosh2019qrp,
  author = {Ghosh, Sanjib and Opala, Andrzej and Matuszewski, Micha{\l} and Paterek, Tomasz and Liew, Timothy C. H.},
  title = {Quantum Reservoir Processing},
  journal = {npj Quantum Information},
  volume = {5},
  pages = {35},
  year = {2019},
  doi = {10.1038/s41534-019-0149-8}
}

@article{govia2021oscillator,
  author = {Govia, L. C. G. and Ribeill, G. J. and Rowlands, G. E. and Krovi, H. K. and Ohki, T. A.},
  title = {Quantum Reservoir Computing with a Single Nonlinear Oscillator},
  journal = {Physical Review Research},
  volume = {3},
  number = {1},
  pages = {013077},
  year = {2021},
  doi = {10.1103/PhysRevResearch.3.013077}
}

@article{senanian2024microwave,
  author = {Senanian, Alen and others},
  title = {Microwave Signal Processing Using an Analog Quantum Reservoir Computer},
  journal = {Nature Communications},
  volume = {15},
  pages = {7490},
  year = {2024},
  doi = {10.1038/s41467-024-51161-8}
}

@article{paparelle2026optical,
  author = {Paparelle, Iris and Henaff, Johan and Garc{\'i}a-Beni, Jorge and Gillet, {\'E}milie and Montesinos, Daniel and Giorgi, Gian Luca and Soriano, Miguel C. and Zambrini, Roberta and Parigi, Valentina},
  title = {Experimental Memory Control in Continuous-Variable Optical Quantum Reservoir Computing},
  journal = {Nature Photonics},
  volume = {20},
  pages = {413--420},
  year = {2026},
  doi = {10.1038/s41566-026-01880-9}
}

@article{kobayashi2024feedback,
  author = {Kobayashi, Kaito and Fujii, Keisuke and Yamamoto, Naoki},
  title = {Feedback-Driven Quantum Reservoir Computing for Time-Series Analysis},
  journal = {PRX Quantum},
  volume = {5},
  number = {4},
  pages = {040325},
  year = {2024},
  doi = {10.1103/PRXQuantum.5.040325}
}

@article{nokkala2021gaussian,
  author = {Nokkala, Johannes and Mart{\'\i}nez-Pe{\~n}a, Rodrigo and Giorgi, Gian Luca and Parigi, Valentina and Soriano, Miguel C. and Zambrini, Roberta},
  title = {Gaussian States of Continuous-Variable Quantum Systems Provide Universal and Versatile Reservoir Computing},
  journal = {Communications Physics},
  volume = {4},
  pages = {53},
  year = {2021},
  doi = {10.1038/s42005-021-00556-w}
}

@article{martinezpena2023ipc,
  author = {Mart{\'\i}nez-Pe{\~n}a, Rodrigo and Nokkala, Johannes and Giorgi, Gian Luca and Zambrini, Roberta and Soriano, Miguel C.},
  title = {Information Processing Capacity of Spin-Based Quantum Reservoir Computing Systems},
  journal = {Cognitive Computation},
  volume = {15},
  pages = {1440--1451},
  year = {2023},
  doi = {10.1007/s12559-020-09772-y}
}

@article{llodra2023particle,
  author = {Llodr{\`a}, Guillem and Charalambous, Christos and Giorgi, Gian Luca and Zambrini, Roberta},
  title = {Benchmarking the Role of Particle Statistics in Quantum Reservoir Computing},
  journal = {Advanced Quantum Technologies},
  volume = {6},
  pages = {2200100},
  year = {2023},
  doi = {10.1002/qute.202200100}
}

@article{cindrak2024memory,
  author = {{\v{C}}indrak, Simon and Donvil, Brecht and L{\"u}dge, Kathy and Jaurigue, Lina},
  title = {Enhancing the Performance of Quantum Reservoir Computing and Solving the Time-Complexity Problem by Artificial Memory Restriction},
  journal = {Physical Review Research},
  volume = {6},
  pages = {013051},
  year = {2024},
  doi = {10.1103/PhysRevResearch.6.013051}
}

@misc{supplemental,
  note = {See the Supplemental Material appended to this manuscript for
  detailed proofs, exact population-moment calculations, and additional
  numerical results}
}

@book{sup_bhatia1997matrix,
  title     = {Matrix Analysis},
  author    = {Bhatia, Rajendra},
  series    = {Graduate Texts in Mathematics},
  volume    = {169},
  publisher = {Springer},
  address   = {New York},
  year      = {1997},
  doi       = {10.1007/978-1-4612-0653-8}
}

@article{sup_collins2006integration,
  title   = {Integration with respect to the {Haar} measure on unitary, orthogonal and symplectic group},
  author  = {Collins, Beno{\^i}t and {\'S}niady, Piotr},
  journal = {Communications in Mathematical Physics},
  volume  = {264},
  number  = {3},
  pages   = {773--795},
  year    = {2006},
  doi     = {10.1007/s00220-006-1554-3}
}

@article{sup_mele2024swap,
  title   = {Introduction to Haar Measure Tools in Quantum Information: A Beginner's Tutorial},
  author  = {Mele, Antonio Anna},
  journal = {Quantum},
  volume  = {8},
  pages   = {1340},
  year    = {2024},
  doi     = {10.22331/q-2024-05-08-1340}
}

@article{sup_hoerl1970ridge,
  title   = {Ridge Regression: Biased Estimation for Nonorthogonal Problems},
  author  = {Hoerl, Arthur E. and Kennard, Robert W.},
  journal = {Technometrics},
  volume  = {12},
  number  = {1},
  pages   = {55--67},
  year    = {1970},
  doi     = {10.1080/00401706.1970.10488634}
}

@article{sup_efronstein1981jackknife,
  title   = {The Jackknife Estimate of Variance},
  author  = {Efron, Bradley and Stein, Charles},
  journal = {The Annals of Statistics},
  volume  = {9},
  number  = {3},
  pages   = {586--596},
  year    = {1981},
  doi     = {10.1214/aos/1176345462}
}

@article{sup_bravyi2005fermionic,
  title   = {Lagrangian Representation for Fermionic Linear Optics},
  author  = {Bravyi, Sergey},
  journal = {Quantum Information and Computation},
  volume  = {5},
  number  = {3},
  pages   = {216--238},
  year    = {2005},
  doi     = {10.26421/QIC5.3-3},
  eprint  = {quant-ph/0404180},
  archivePrefix = {arXiv}
}

@article{sup_almohy2009frechet,
  title   = {Computing the Fr\'{e}chet Derivative of the Matrix Exponential,
             with an Application to Condition Number Estimation},
  author  = {Al-Mohy, Awad H. and Higham, Nicholas J.},
  journal = {SIAM Journal on Matrix Analysis and Applications},
  volume  = {30},
  number  = {4},
  pages   = {1639--1657},
  year    = {2009},
  doi     = {10.1137/080716426}
}

@article{sup_mityagin2020zero,
  title   = {The Zero Set of a Real Analytic Function},
  author  = {Mityagin, Boris S.},
  journal = {Mathematical Notes},
  volume  = {107},
  number  = {3--4},
  pages   = {529--530},
  year    = {2020},
  doi     = {10.1134/S0001434620030189}
}

@article{sup_efron1979bootstrap,
  title   = {Bootstrap Methods: Another Look at the Jackknife},
  author  = {Efron, Bradley},
  journal = {The Annals of Statistics},
  volume  = {7},
  number  = {1},
  pages   = {1--26},
  year    = {1979},
  doi     = {10.1214/aos/1176344552}
}

@article{sup_garoni2005density,
  author  = {Garoni, Timothy M. and Forrester, Peter J. and Frankel, Norman E.},
  title   = {Asymptotic corrections to the eigenvalue density of the {GUE} and {LUE}},
  journal = {Journal of Mathematical Physics},
  volume  = {46},
  pages   = {103301},
  year    = {2005},
  doi     = {10.1063/1.2035028},
  eprint  = {math-ph/0504053},
  archivePrefix = {arXiv}
}

\clearpage
\onecolumngrid
\setcounter{theorem}{0}
\setcounter{equation}{0}
\setcounter{section}{0}
\setcounter{figure}{0}
\setcounter{table}{0}
\setcounter{secnumdepth}{3}
\renewcommand{\thetheorem}{S\arabic{theorem}}
\renewcommand{\theequation}{S\arabic{equation}}
\renewcommand{\thesection}{S\arabic{section}}
\renewcommand{\thefigure}{S\arabic{figure}}
\renewcommand{\thetable}{S\arabic{table}}
\renewcommand{\theHtheorem}{S\arabic{theorem}}
\renewcommand{\theHequation}{S\arabic{equation}}
\renewcommand{\theHsection}{S\arabic{section}}
\renewcommand{\theHfigure}{S\arabic{figure}}
\renewcommand{\theHtable}{S\arabic{table}}

\begin{center}
{\large\bfseries Supplemental Material for \texorpdfstring{\\}{ }
Theory of the Edge-of-Chaos Advantage in Quantum Reservoir Computing}
\end{center}

This Supplemental Material provides the definitions, proofs, and numerical
details supporting the theoretical results and crossover calculation in the main
text.
Section~\ref{sm:sec:v2-protocol} establishes the common QRC and regression
framework.  Section~\ref{sm:sec:v2-forgetting} derives the forgetting bounds
in the chaotic regime.  Section~\ref{sm:sec:v2-free} proves the free-fermion
obstruction and establishes interaction activation under the QRC rewinding protocol.
Section~\ref{sm:sec:v2-two-copy} derives the exact stationary population moments
used across the quadratic--quartic crossover, and
Sec.~\ref{sm:sec:v2-numerics} gives additional numerical results.

\section{Protocol and population ridge regression}
\label{sm:sec:v2-protocol}

Let the reservoir Hilbert space be
$\mathcal H_A\otimes\mathcal H_{\bar{A}}$, where $A$ is the reset subsystem and
$\bar{A}$ is the memory subsystem,
$d_A=\dim\mathcal H_A$, $d_{\bar{A}}=\dim\mathcal H_{\bar{A}}$, and $d=d_Ad_{\bar{A}}$.
We keep $d_A$ explicit in intermediate formulas that also apply to a general
reset subsystem.  The reservoir studied in the main text has $d_A=2$.
At step $k$, the input $u_k$ is encoded in a density operator
$\rho_{\rm in}(u_k)$ on $A$.  The reservoir then evolves with
$U=\exp(-\im Ht)$, where $H$ is its fixed Hermitian Hamiltonian and $t$
is the evolution time between successive inputs, with $\hbar=1$.
Let $\rho_k$ be the resulting full reservoir state and
$\sigma_k:=\Tr_A\rho_k$ its memory state, where $\Tr_A$ denotes the
partial trace over $A$.  One cycle of reset and evolution is described by
\begin{align}
    \rho_k
    &=U[\rho_{\rm in}(u_k)\otimes\sigma_{k-1}]U^\dagger,
    \label{sm:eq:v2-reset-full}\\
    \sigma_k
    &=\Phi_{u_k,U}(\sigma_{k-1}),
    &
    \Phi_{u,U}(\sigma)
    &:=\Tr_A\{U[\rho_{\rm in}(u)\otimes\sigma]U^\dagger\}.
    \label{sm:eq:v2-reset-memory}
\end{align}
The same linear map $\Phi_{u,U}$ also acts on differences of memory states.
The inputs are independent and identically distributed (i.i.d.) unless explicitly
stated otherwise.  Except for the finite-depth interaction-activation result
in Sec.~\ref{sm:sec:v2-activation}, the reservoir is assumed to satisfy the
echo-state property (ESP), so that Eq.~\eqref{sm:eq:v2-reset-memory} defines a
unique stationary causal process~\cite{chen2020temporal}.  The finite-depth
result instead specifies a fixed input-independent initialization.
Unless stated otherwise, $\mathbb E_{\rm in}$, $\operatorname{Cov}$, and
$\operatorname{Var}$ denote expectation, covariance, and variance over the
input law of the process under consideration with $U$ fixed.  For two
QRC runs whose inputs differ only at one step, let $u'$ denote the
input independently drawn from the input distribution at that step in the
second run. 
$\mathbb E_{\mathrm{in}'}$ denotes the expectation over the input history
of the first run and $u'$.  Fresh-input and Hamiltonian-ensemble averages are
identified explicitly.

For the product task studied below, $A$ is the first qubit and we use
amplitude encoding for $u\in[0,1]$:
\begin{equation}
    \rho_{\rm in}(u)
    =|\psi(u)\rangle\!\langle\psi(u)|,
    \qquad
    |\psi(u)\rangle=\sqrt{1-u}|0\rangle+\sqrt u|1\rangle.
    \label{sm:eq:v2-amplitude-encoding}
\end{equation}
Equivalently, with $z(u):=u-1/2$, $z_k:=z(u_k)$, and $\phi(u):=2\sqrt{u(1-u)}$,
\begin{equation}
    \rho_{\rm in}(u)=\frac I2+\frac{\phi(u)}2X-z(u)Z.
    \label{sm:eq:v2-amplitude-bloch}
\end{equation}
Here $X$, $Y$, and $Z$ denote the Pauli operators on $A$.  We add the
subscript $A$ when displaying tensor products with the memory subsystem.
We write $I$ for an identity operator or matrix of the dimension required
by its expression, and $I_A$ and $I_{\bar A}$ when specifying the subsystem.

Let $O_1,\ldots,O_m$ be the $m$ Hermitian readout observables on the full
reservoir, normalized by $\Tr(O_iO_j)=\delta_{ij}$.  The feature vector is
$\bm x_k=(x_{k,1},\ldots,x_{k,m})^\top$, where
$x_{k,i}:=\Tr(O_i\rho_k)$ for $i=1,\ldots,m$.
For an operator $B$, the squared unnormalized Hilbert--Schmidt norm is
$\|B\|_2^2=\Tr(B^\dagger B)$.  Unsubscripted vector norms are Euclidean norms.
A random variable is called centered when its expectation is zero.
For a real target $y_k$ with variance
$\sigma_y^2:=\operatorname{Var}(y_k)\in(0,\infty)$, define the feature
covariance matrix and feature--target covariance vector
\begin{equation}
    \Sigma:=\operatorname{Cov}(\bm x_k,\bm x_k),
    \qquad
    \bm c:=\operatorname{Cov}(\bm x_k,y_k).
    \label{sm:eq:v2-population-moments}
\end{equation}
The $i$th component of $\bm c$ is
$c_i:=\operatorname{Cov}(x_{k,i},y_k)$.
Population ridge regression~\cite{sup_hoerl1970ridge} fits the intercept
$b\in\mathbb R$ and readout weights $\bm w\in\mathbb R^m$ by minimizing
\begin{equation}
    \mathbb E_{\rm in}[(y_k-b-\bm w^\top\bm x_k)^2]+\lambda\|\bm w\|^2,
    \qquad \lambda>0.
    \label{sm:eq:v2-population-ridge}
\end{equation}
Here $\lambda>0$ is the ridge regularization parameter.  The fitted weights are
$\bm w_\lambda=K_\lambda\bm c$, where
$K_\lambda=(\Sigma+\lambda I)^{-1}$, and the fitted intercept is
$b_\lambda=\mathbb E_{\rm in}y_k
-\bm w_\lambda^\top\mathbb E_{\rm in}\bm x_k$.
The population normalized mean-squared error (NMSE), evaluated without the
penalty term, is
\begin{equation}
    \operatorname{NMSE}_{\rm pop}
    :=\frac{\mathbb E_{\rm in}[(y_k-b_\lambda-\bm w_\lambda^\top\bm x_k)^2]}
    {\sigma_y^2}.
    \label{sm:eq:v2-population-nmse}
\end{equation}
The corresponding performance is
\begin{align}
    P_\lambda
    &:=1-\operatorname{NMSE}_{\rm pop}=\frac{\bm c^{\mathsf T}
      (K_\lambda+\lambda K_\lambda^2)\bm c}
      {\sigma_y^2}.
    \label{sm:eq:v2-exact-ridge-gain}
\end{align}
Indeed,
$K_\lambda\Sigma K_\lambda=K_\lambda-\lambda K_\lambda^2$.

The analytic formulas above use Hilbert--Schmidt-normalized observables and
the penalty $\lambda$.  In the numerical calculations we instead use the
physical Pauli or Majorana bilinears $\widetilde O_i$, with
$\Tr(\widetilde O_i^2)=d$, as raw feature observables.
Let $\lambda_{\rm raw}$ be the ridge regularization parameter for these observables.
The subscripts ``raw'' and ``norm'' refer to quantities computed using
$\widetilde O_i$ and $O_i$, respectively.  Thus
$\Sigma_{\rm norm}=\Sigma$ and $\bm c_{\rm norm}=\bm c$.
Since $O_i=\widetilde O_i/\sqrt d$, the covariances and penalties satisfy
\begin{equation}
    \Sigma_{\rm norm}=\frac{\Sigma_{\rm raw}}d,
    \qquad
    \bm c_{\rm norm}=\frac{\bm c_{\rm raw}}{\sqrt d},
    \qquad
    \lambda_{\rm norm}=\frac{\lambda_{\rm raw}}d.
    \label{sm:eq:v2-coordinate-scaling}
\end{equation}
Thus $\lambda_{\rm raw}$ and $\lambda_{\rm norm}$ represent the same ridge
regularization in two feature-coordinate conventions, and the rescaling
leaves $P_\lambda$ unchanged.  Hereafter $\lambda$ denotes
$\lambda_{\rm norm}$ in normalized formulas.  Section~\ref{sm:sec:v2-numerics}
reports penalties in the raw convention and uses
Eq.~\eqref{sm:eq:v2-coordinate-scaling} when applying those formulas.

\section{Chaotic regime: GUE forgetting under the EDA}
\label{sm:sec:v2-forgetting}

For a square-integrable target $y_k$ and an integer delay $\tau\geq1$, define the
remaining-input sigma-algebra
\begin{equation}
    \mathcal F_{k,\tau}
    :=\sigma(u_j:j\leq k,\ j\neq k-\tau).
\end{equation}
Here $\sigma(\cdot)$ denotes the sigma-algebra generated by the enclosed
random variables, so $\mathcal F_{k,\tau}$ contains all input information
through time $k$ except $u_{k-\tau}$.
We call $y_k$ a lag-$\tau$ innovation target when
\begin{equation}
    \mathbb E_{\rm in}[y_k\mid\mathcal F_{k,\tau}]
    =\mathbb E_{\rm in} y_k.
    \label{sm:eq:v2-lag-innovation}
\end{equation}
This equality holds with probability one over the input histories.  The best
prediction from the other inputs alone is therefore the target mean.
Examples include any centered function of $u_{k-\tau}$ and its product
with a centered function of $u_k$.

Fix $U$ and the observation time $k$.  In this section, we display the
unitary dependence by writing $\rho_{k,U}:=\rho_k$, $P_\lambda(U)$,
and $\bm c(U)$ for the state, performance, and covariance vector
at this fixed $U$.  Two QRC runs use the same input
history through time $k-\tau-1$.  At time $k-\tau$, the first run receives the
input $u_{k-\tau}$ and the second receives $u'_{k-\tau}$, drawn
from the same input distribution independently of $U$ and the entire
input history of the first run.  All other inputs, including those at times
$k-\tau+1,\ldots,k$, are identical in the two runs.  The unitary $U$
also remains the same.  Let $\rho_{k,U}^{(k-\tau)}$ denote the state of the
second run at time $k$.  The expectation $\mathbb E_{\mathrm{in}'}$ averages
over the input history of the first run and the independently drawn $u'_{k-\tau}$,
with $U$ held fixed.  Define
\begin{equation}
    \mathcal I_\tau(U)
    :=\frac12\mathbb E_{\mathrm{in}'}
      \left\|\rho_{k,U}-\rho_{k,U}^{(k-\tau)}\right\|_2^2.
    \label{sm:eq:v2-lag-influence}
\end{equation}
Conditional on $\mathcal F_{k,\tau}$, the two states are i.i.d. because
their different inputs are drawn independently.  Let
\begin{equation}
    \overline\rho_{k,U}
    :=\mathbb E_{\rm in}[\rho_{k,U}\mid\mathcal F_{k,\tau}]
\end{equation}
be the conditional mean state, obtained by averaging over $u_{k-\tau}$
while holding all other inputs fixed.  Expanding the squared difference
and using conditional independence to remove
the cross term gives the conditional independent-copy
variance identity~\cite{sup_efronstein1981jackknife},
\begin{equation}
\begin{aligned}
 &\mathbb E_{\mathrm{in}'}\!\left[\left.
 \left\|\rho_{k,U}-\rho_{k,U}^{(k-\tau)}\right\|_2^2
 \right|\mathcal F_{k,\tau}\right]\\
 &\quad=\mathbb E_{\mathrm{in}'}\!\left[\left.
 \left\|\rho_{k,U}-\overline\rho_{k,U}\right\|_2^2
 \right|\mathcal F_{k,\tau}\right]
 +\mathbb E_{\mathrm{in}'}\!\left[\left.
 \left\|\rho_{k,U}^{(k-\tau)}-\overline\rho_{k,U}\right\|_2^2
 \right|\mathcal F_{k,\tau}\right]\\
 &\quad=2\mathbb E_{\mathrm{in}'}\!\left[\left.
 \left\|\rho_{k,U}-\overline\rho_{k,U}\right\|_2^2
 \right|\mathcal F_{k,\tau}\right].
\end{aligned}
\end{equation}
Combining this identity with Eq.~\eqref{sm:eq:v2-lag-influence} and averaging
over $\mathcal F_{k,\tau}$ yields
\begin{equation}
    \mathcal I_\tau(U)
    =\mathbb E_{\rm in}\|\rho_{k,U}-\overline\rho_{k,U}\|_2^2.
    \label{sm:eq:v2-influence-variance}
\end{equation}

The following lemma bounds the performance without requiring
$\Sigma$ and converts the surviving input influence into the performance bound
in the chaotic regime.

\begin{lemma}[Sharp covariance-free ridge envelope]
\label{sm:lem:v2-sharp-envelope}
Let $\bm a$ be the feature--target covariance divided by the target standard
deviation, and let $\mathcal R$ be its squared norm.  Define
\begin{equation}
    \bm a:=\frac{\bm c}{\sigma_y},
    \qquad
    \mathcal R:=\|\bm a\|^2,
    \qquad
    f_\lambda(r):=1-\left(\frac{\lambda}{\lambda+r}\right)^2,\quad r\geq0.
    \label{sm:eq:v2-psi-definition}
\end{equation}
For every joint feature--target distribution with finite second moments,
\begin{equation}
    P_\lambda\leq f_\lambda(\mathcal R)
    \leq\min\left\{1,\frac{2\mathcal R}{\lambda}\right\}.
    \label{sm:eq:v2-sharp-envelope}
\end{equation}
For fixed $\mathcal R$ and $\lambda$, the first bound can be attained by a
joint feature--target distribution.  If
$\mathcal R>0$, equality holds precisely when
\begin{equation}
    (\Sigma-\bm a\bm a^\top)\bm a=0.
    \label{sm:eq:v2-sharp-equality}
\end{equation}
\end{lemma}

\textit{Proof.}
The covariance matrix of $\bm x_k-\mathbb E_{\rm in}\bm x_k$ and
$(y_k-\mathbb E_{\rm in}y_k)/\sigma_y$ is positive semidefinite, so its
Schur complement gives
\begin{equation}
    \Sigma_\perp:=\Sigma-\bm a\bm a^\top\succeq0.
\end{equation}
For this proof, use the weights
$\bm w_{\rm std}:=K_\lambda\bm a=\bm w_\lambda/\sigma_y$ associated with
the target divided by its standard deviation, and set
$s:=\bm a^\top\bm w_{\rm std}$.  Equation~\eqref{sm:eq:v2-exact-ridge-gain} gives
\begin{equation}
    P_\lambda
    =2s-\bm w_{\rm std}^\top\Sigma\bm w_{\rm std}
    =2s-s^2-\bm w_{\rm std}^\top\Sigma_\perp\bm w_{\rm std}
    \leq2s-s^2.
    \label{sm:eq:v2-sharp-proof-one}
\end{equation}
Since
$\Sigma+\lambda I\succeq\bm a\bm a^\top+\lambda I$,
inverse order and the Sherman--Morrison identity imply
\begin{equation}
    0\leq s\leq
    \bm a^\top(\bm a\bm a^\top+\lambda I)^{-1}\bm a
    =\frac{\mathcal R}{\mathcal R+\lambda}<1.
\end{equation}
These standard positive-matrix and rank-one inverse identities are reviewed
in Ref.~\cite{sup_bhatia1997matrix}.
The function $2s-s^2$ is increasing on $[0,1]$, which proves the first
inequality in Eq.~\eqref{sm:eq:v2-sharp-envelope}; the second is elementary.
Equality in the inverse-order bound occurs precisely when
$\Sigma_\perp(\bm a\bm a^\top+\lambda I)^{-1}\bm a=0$, equivalently when
$\Sigma_\perp\bm a=0$.  In that case
$\bm w_{\rm std}=\bm a/(\mathcal R+\lambda)$ and
$\bm w_{\rm std}^\top\Sigma_\perp\bm w_{\rm std}=0$, so Eq.~\eqref{sm:eq:v2-sharp-proof-one} is also
saturated.  Conversely, equality in the envelope forces equality in these two
steps and hence the same condition.  Taking
$\Sigma=\bm a\bm a^\top$ establishes optimality at fixed $\mathcal R$.
\hfill$\blacksquare$

\begin{lemma}[Performance bound from input influence]
\label{sm:thm:v2-universal-forgetting}
For i.i.d. inputs, every fixed $U$ admitting a causal stationary process,
every lag-$\tau$ innovation target, and every $\lambda>0$,
\begin{equation}
    P_\lambda(U)
    \leq f_\lambda[\mathcal I_\tau(U)]
    \leq\min\left\{1,
      \frac{2\mathcal I_\tau(U)}{\lambda}
      \right\}.
    \label{sm:eq:v2-universal-forgetting}
\end{equation}
\end{lemma}

\textit{Proof.}
Write $y_k^\circ:=y_k-\mathbb E_{\rm in}y_k$ for the centered target.
Equation~\eqref{sm:eq:v2-lag-innovation} implies
$\mathbb E_{\rm in}[y_k^\circ\mid\mathcal F_{k,\tau}]=0$.
The feature--target covariance can therefore be expressed through the
target--state covariance operator $C_U$:
\begin{equation}
    c_i(U)=\Tr(O_iC_U),
    \qquad
    C_U:=\mathbb E_{\rm in}\{y_k^\circ
      (\rho_{k,U}-\overline\rho_{k,U})\}.
    \label{sm:eq:v2-influence-covariance-operator}
\end{equation}
Orthonormality of the observables, Hilbert-space Cauchy--Schwarz, and
Eq.~\eqref{sm:eq:v2-influence-variance} give
\begin{equation}
    \|\bm c(U)\|^2
    \leq\|C_U\|_2^2
    \leq\sigma_y^2\mathcal I_\tau(U).
    \label{sm:eq:v2-covariance-influence-bound}
\end{equation}
Apply Lemma~\ref{sm:lem:v2-sharp-envelope} and use the monotonicity of
$f_\lambda$.
\hfill$\blacksquare$

We now follow the input influence through successive cycles of reset and
evolution.  At fixed evolution time $t$, choose an ensemble of Hamiltonians $H$ and
let $\nu_t$ be the induced distribution of $U(t)=\exp(-\im Ht)$.  For each $U\sim\nu_t$, compare the two runs defined
above, whose inputs differ only at step $k-\tau$.
We write $\mathbb E_{U\sim\nu_t}$ for expectation over this draw.
Their memory difference is
\begin{equation}
 \Delta_j:=\Tr_A\!\left[
 \rho_{k-\tau+j,U}-\rho_{k-\tau+j,U}^{(k-\tau)}
 \right],\qquad j=0,\ldots,\tau-1.
\end{equation}
The subscript on $\rho$ gives the time at which the state is observed,
after that step's reset and evolution.  Thus $\Delta_0$ is the difference immediately after
the two different inputs at time $k-\tau$ have each been encoded and evolved.
The transition from $\Delta_j$ to $\Delta_{j+1}$ is one further round of
reset and evolution, with the same next input in both runs.

The map $\Phi_{u,U}$ in Eq.~\eqref{sm:eq:v2-reset-memory} acts linearly on
traceless Hermitian operators $\Delta$ on the memory subsystem $\bar A$:
\begin{equation}
 \Phi_{u,U}(\Delta)
 =\Tr_A\!\left[U(\rho_{\rm in}(u)\otimes\Delta)U^\dagger\right].
\end{equation}
For a memory-state difference, its output is the difference remaining after
both runs receive the common input $u$ and evolve with the same $U$.
For $j=0,\ldots,\tau-2$,
\begin{equation}
 \Delta_{j+1}=\Phi_{u_{k-\tau+j+1},U}(\Delta_j).
\end{equation}
The next common input is $u_{k-\tau+j+1}$.  We denote it simply by $u$ when
averaging this one update.  Because the inputs are i.i.d., it is independent
of the inputs already used in the two runs.  For fixed $U$ and fixed
nonzero $\Delta$, define the retention factor per step by
\begin{equation}
 q_U(\Delta):=\frac{\mathbb E_u\|\Phi_{u,U}(\Delta)\|_2^2}{\|\Delta\|_2^2},
 \label{sm:eq:v2-q-definition}
\end{equation}
where $\mathbb E_u$ averages only this next common input, drawn from the
input distribution.  Thus $q_U(\Delta)$ is the mean squared difference
after one update divided by the squared difference before it.  We set
$q_U(0)=0$.

For averages over the input histories and reservoir ensemble, write
$\mathbb E_{\rm joint}:=\mathbb E_{U\sim\nu_t}\mathbb E_{\mathrm{in}'}$:
we first average the input history of the first run and the independently drawn
$u'_{k-\tau}$ at fixed $U$, and then average over $U\sim\nu_t$.
Conditioning on $U$ and the inputs
already used, and averaging the next common input, gives the exact identity
\begin{equation}
 \mathbb E_{\rm joint}\|\Delta_{j+1}\|_2^2
 =\mathbb E_{\rm joint}[\|\Delta_j\|_2^2q_U(\Delta_j)].
 \label{sm:eq:v2-exact-influence-recursion}
\end{equation}
This identity explains the weight $\|\Delta_j\|_2^2$ in the definition of EDA: for the same retention
factor, a pair of runs with a larger remaining difference contributes more
to the mean squared difference after the next update.

For the influence propagation below, assume pure reset states, as in the
main text.  At the observation time $k$, the memory difference before the final reset
is $\Delta_{\tau-1}$.  The full-state difference after that reset and
evolution is
\begin{equation}
 \rho_{k,U}-\rho_{k,U}^{(k-\tau)}
 =U\bigl[\rho_{\rm in}(u_k)\otimes \Delta_{\tau-1}\bigr]U^\dagger.
\end{equation}
Unitary invariance of the Hilbert--Schmidt norm and the purity of the reset
state, $\Tr[\rho_{\rm in}(u_k)^2]=1$, therefore give
$\|\rho_{k,U}-\rho_{k,U}^{(k-\tau)}\|_2^2
=\|\Delta_{\tau-1}\|_2^2$.  Hence
\begin{equation}
 \overline{\mathcal I}_\tau
 :=\mathbb E_{U\sim\nu_t}\mathcal I_\tau(U)
 =\frac12\mathbb E_{U\sim\nu_t}
   \mathbb E_{\mathrm{in}'}\|\Delta_{\tau-1}\|_2^2.
 \label{sm:eq:v2-ensemble-influence}
\end{equation}

We now introduce the effective decorrelation approximation (EDA).
First, define the ensemble mean of the
scalar-valued function $q_U$ by
\begin{equation}
 \bar q_t:=\mathbb E_{U\sim\nu_t}[q_U].
\end{equation}
This is a pointwise average: we hold the memory-operator argument fixed
while averaging the values of $q_U$ over the reservoir ensemble.  We then
evaluate the resulting deterministic function at the actual memory
difference $\Delta_j$. 

The value $\bar q_t(\Delta_j)$ gives the mean retention if the unitary for
the next cycle were drawn independently of $\Delta_j$.  However, the original $U$
can be correlated with $\Delta_j$, since it was used to generate that
difference.  The EDA uses
$\bar q_t$ to estimate the mean squared memory difference after the next
cycle.  Using the exact recursion in
Eq.~\eqref{sm:eq:v2-exact-influence-recursion}, its condition for every step
$j=0,\ldots,\tau-2$ in the chosen window can be written as
\begin{equation}
 \frac{\left|
 \mathbb E_{\rm joint}\|\Delta_{j+1}\|_2^2
 -\mathbb E_{\rm joint}[\|\Delta_j\|_2^2\bar q_t(\Delta_j)]
 \right|}
 {\mathbb E_{\rm joint}\|\Delta_j\|_2^2}
 \leq\varepsilon_{\rm EDA}.
 \label{sm:eq:v2-eda}
\end{equation}
The condition applies when the denominator is positive.  If it vanishes,
all subsequent memory differences vanish as well.
The error is measured relative to the current mean squared difference.
The same $\varepsilon_{\rm EDA}\geq0$ controls all steps in this window
and obeys $\varepsilon_{\rm EDA}\to0$ as $d_{\bar A}\to\infty$ at fixed
reset subsystem $A$.  It may depend on the chosen window; no uniform
claim at arbitrarily long delays is assumed.  In the main text $A$
is one qubit, so $d_A=2$.

To evaluate $\bar q_t$ in the chaotic regime, we draw Hamiltonians from
the normalized Gaussian unitary ensemble (GUE)~\cite{cotler2017chaos_designs},
\begin{equation}
    \mathrm d\mathbb P_d(H)
    =Z_d^{-1}\exp[-(d/2)\Tr(H^2)]\,\mathrm dH,
    \qquad U(t)=\e^{-\im Ht}.
    \label{sm:eq:v2-normalized-gue}
\end{equation}
Here $H$ ranges over $d\times d$ Hermitian matrices, $\mathrm dH$ is the
product measure over their independent real components, and $Z_d$ is the
normalization constant.  Let $\mathbb E_H$ denote expectation over the GUE draw in
Eq.~\eqref{sm:eq:v2-normalized-gue}.
We consider evolution times $t\geq\sqrt d$.

\begin{theorem}[GUE forgetting under the EDA]
\label{sm:thm:v2-gue-forgetting}
Let $\tau\geq1$ be an integer and $y_k$ a lag-$\tau$ innovation target.
For GUE reservoirs with a single reset qubit and evolution time
$t\geq\sqrt d$, if the EDA holds for $j=0,\ldots,\tau-2$
(an empty condition for $\tau=1$), then
\begin{equation}
 \mathbb E_{U\sim\nu_t}P_\lambda(U)
 \leq\frac{2\overline{\mathcal I}_1}{\lambda}
 \left[\frac12+o(1)\right]^{\tau-1}.
 \label{sm:eq:v2-gue-main-performance-bound}
\end{equation}
Here $o(1)$ tends to zero as $d_{\bar A}\to\infty$.
The performance bound decays exponentially with lag over the range covered
by the EDA.  
\end{theorem}

\textit{Proof.}
We first calculate the GUE mean retention factor, then use the EDA to bound
the surviving influence and hence the performance.

\textit{1. Mean retention from the GUE spectral form factors.}
We first establish that the mean retention is the same for every nonzero
memory difference and every pure reset state.
Let $\rho_{\rm in}$ be pure, let $\Delta$ be a nonzero traceless Hermitian memory
operator independent of $H$, and put $B=\rho_{\rm in}\otimes\Delta$.
For the GUE, the ensemble-averaged ratio is independent of both $\Delta$
and the choice of pure reset state, as shown below.
We therefore suppress the argument of the average function and write
its common value as $\bar q_t$:
\begin{equation}
    \bar q_t=
    \mathbb E_H\frac{
      \|\Tr_A[U(t)B U(t)^\dagger]\|_2^2}{\|\Delta\|_2^2}.
    \label{sm:eq:v2-gue-contraction}
\end{equation}
To prove this, for a fixed pure reset state $\rho$ and any Hermitian memory
operator $\Delta$, define the quadratic form
\begin{equation}
 F_t^{\rho}(\Delta):=\mathbb E_H
 \left\|\Tr_A\!\left[U(t)(\rho\otimes \Delta)U(t)^\dagger\right]\right\|_2^2.
\end{equation}
For any memory unitary $W_{\bar{A}}$, put
$W=I_A\otimes W_{\bar{A}}$.  The matrices $H$ and $W^\dagger HW$ have the same GUE distribution.
Together with
$\e^{-\im Ht}W=W\e^{-\im(W^\dagger HW)t}$, this gives
\begin{equation}
 F_t^{\rho}(W_{\bar{A}}\Delta W_{\bar{A}}^\dagger)=F_t^{\rho}(\Delta).
\end{equation}
Let $\delta_1,\ldots,\delta_{d_{\bar A}}$ be the eigenvalues of $\Delta$.
The invariance above allows us to diagonalize $\Delta$ without changing
$F_t^\rho$.  This function is quadratic in the eigenvalues and unchanged
by any permutation of them.  Consequently, all terms $\delta_a^2$ have
one common coefficient, and all products $\delta_a\delta_b$ with $a\ne b$
have another.  Equivalently, for some real coefficients $\eta_1(t)$ and
$\eta_0(t)$ independent of $\Delta$, we can write
\begin{align}
 F_t^{\rho}(\Delta)
 &=\eta_1(t)\sum_{a=1}^{d_{\bar A}}\delta_a^2
   +\eta_0(t)\left(\sum_{a=1}^{d_{\bar A}}\delta_a\right)^2\nonumber\\
 &=\eta_1(t)\Tr(\Delta^2)+\eta_0(t)[\Tr\Delta]^2.
\end{align}
For a nonzero memory difference, $\Tr\Delta=0$, so
$F_t^\rho(\Delta)/\|\Delta\|_2^2=\eta_1(t)$.
Thus the ensemble-averaged retention factor is the same for every nonzero
memory difference.
Any two pure reset states are related by
$\rho'=W_A\rho W_A^\dagger$ for some unitary $W_A$ on $A$.
The same change of GUE variable with
$W=W_A\otimes I_{\bar{A}}$ leaves the output partial trace, and therefore its norm,
unchanged.  Thus $F_t^\rho$ is also independent of the pure reset state $\rho$, and averaging over
the pure amplitude encoding does not alter the ratio.

It remains to calculate this common retention factor.  We first average
over eigenvectors while holding the eigenvalues fixed, and then average
over the eigenvalues.
Write $H=V\operatorname{diag}(E_1,\ldots,E_d)V^\dagger$, where the columns
of $V$ are the eigenvectors and $\bm E:=(E_1,\ldots,E_d)$ is the spectrum.
With $\Lambda(t):=\operatorname{diag}(\e^{-\im E_1t},\ldots,
\e^{-\im E_dt})$, the evolution is $U(t)=V\Lambda(t)V^\dagger$.
For the GUE, $V$ is uniformly distributed over unitary matrices
(Haar measure) and independent of $\bm E$.
Let $\mathbb E_V$ denote this eigenvector average at fixed $\bm E$, and define
\begin{equation}
 q_t^{\rm spec}(\bm E)
 :=\frac{\mathbb E_V
 \|\Tr_A[V\Lambda(t)V^\dagger B
 V\Lambda(t)^\dagger V^\dagger]\|_2^2}{\|\Delta\|_2^2}.
 \label{sm:eq:v2-fixed-spectrum-average-definition}
\end{equation}
The invariance argument above also applies to this average, so the result
is independent of the nonzero memory difference and the pure reset state.
The ratio is unchanged by scaling $\Delta$.  For the calculation below,
we therefore take $\Tr(\Delta^2)=1$, keeping
$B=\rho_{\rm in}\otimes\Delta$.
For a fixed spectrum, write
$s_n(t):=\Tr[\Lambda(t)^n]=\sum_{a=1}^d\e^{-\im nE_at}$ for $n=1,2$.
These are the traces of the evolution at times $t$ and $2t$, before any
average over eigenvalues.

We now express the squared partial-trace norm in a form to which the
standard Haar integral applies.
Let $S_{\bar A}$ swap the two copies of $\bar A$ in two reservoir copies,
leaving both copies of $A$ unchanged.  The partial-trace swap
identity~\cite{sup_mele2024swap} gives
\begin{equation}
 \|\Tr_A[U(t)BU(t)^\dagger]\|_2^2
 =\Tr[(U(t)BU(t)^\dagger)^{\otimes2}S_{\bar A}].
 \label{sm:eq:v2-retention-swap}
\end{equation}
At fixed $\bm E$, we apply the fourth-order Haar integration
formula~\cite{sup_collins2006integration}, using the fixed-spectrum
expression in Appendix~A of
Ref.~\cite{oliviero2021isospectral_twirling}.
The partial trace is handled as in Appendix~F.1 of that reference.
For $B=\rho_{\rm in}\otimes\Delta$, using $\Tr\Delta=0$, the purity of
$\rho_{\rm in}$, and our normalization of $\Delta$ gives $\Tr B=0$, $\Tr(B^2)=1$,
$\Tr_A B=\Delta$, and $\Tr_{\bar A}B=0$.
Substituting these identities into the Haar formula and collecting the
terms involving $s_1(t)$ and $s_2(t)$ gives, for $d>3$,
\begin{align}
 q_t^{\rm spec}(\bm E)
 &=\gamma^2+\frac{d_A-1}{d^2(d^2-1)(d^2-9)}
 \Big\{(d_Ad_{\bar A}^2+3)
 [|s_1(t)|^4+|s_2(t)|^2-4|s_1(t)|^2]\nonumber\\
 &\hspace{35mm}
 -2d_{\bar A}(d_A+3)\operatorname{Re}
 [s_2(t)(s_1(t)^*)^2]\Big\},
 \label{sm:eq:v2-gue-contraction-fixed-spectrum}
\end{align}
where
$\gamma^2:=d_A(d_{\bar A}^2-1)/(d^2-1)$
is the term independent of the spectral phases.
This completes the eigenvector average.

We now average over the GUE eigenvalues, denoting this expectation by
$\mathbb E_{\bm E}$.
The spectral products in Eq.~\eqref{sm:eq:v2-gue-contraction-fixed-spectrum}
give the two-point and four-point spectral form factors
\begin{align}
 \mathcal R_2(t)&:=\mathbb E_{\bm E}|s_1(t)|^2,\qquad
 \mathcal R_4(t):=\mathbb E_{\bm E}|s_1(t)|^4,\nonumber\\
 \mathcal R_{4,1}(t)&:=
 \operatorname{Re}\mathbb E_{\bm E}[s_2(t)(s_1(t)^*)^2],
 \qquad
 \mathcal R_{4,2}(t):=\mathbb E_{\bm E}|s_2(t)|^2=\mathcal R_2(2t).
 \label{sm:eq:v2-gue-form-factor-definitions}
\end{align}
These are ensemble moments of the spectral traces, following the
conventions of Ref.~\cite{cotler2017chaos_designs}.
Taking $\bar q_t=\mathbb E_{\bm E}q_t^{\rm spec}(\bm E)$ therefore
replaces each spectral product by its corresponding form factor:
\begin{align}
    \bar q_t
    &=\gamma^2+
      \frac{d_A-1}{d^2(d^2-1)(d^2-9)}
      \Big\{(d_Ad_{\bar A}^2+3)
      [\mathcal R_4(t)+\mathcal R_{4,2}(t)-4\mathcal R_2(t)]
      \nonumber\\
    &\hspace{42mm}
      -2d_{\bar A}(d_A+3)\mathcal R_{4,1}(t)\Big\}.
    \label{sm:eq:v2-exact-gue-contraction}
\end{align}
We check this exact expression at zero time and in the long-time limit.
At $t=0$, $s_1(0)=s_2(0)=d$, and substitution gives $\bar q_0=1$.
For the long-time check, keep $d$ fixed and take $t\to\infty$.
The GUE eigenvalue distribution has an integrable joint density, so the
ensemble average of a phase term with a nonzero linear combination of
eigenvalues in its exponent vanishes in this limit by the
Riemann--Lebesgue lemma.
For $\mathcal R_2$, the $d$ terms with equal eigenvalue indices remain.
For $\mathcal R_4$, there are two ways to match the two positive and two
negative phase indices.  Each gives $d^2$ terms, and the $d$ terms with
all indices equal are counted twice, leaving $2d^2-d$.
For $\mathcal R_{4,1}$, only the $d$ terms with all three indices equal
remain.  Thus
\begin{equation}
 \mathcal R_2(t),\mathcal R_{4,1}(t),\mathcal R_{4,2}(t)\to d,
 \qquad \mathcal R_4(t)\to2d^2-d
 \qquad (t\to\infty,\ d\ \text{fixed}).
 \label{sm:eq:v2-gue-form-factor-plateau}
\end{equation}
Writing $\bar q_\infty:=\lim_{t\to\infty}\bar q_t$ for the long-time
retention factor, substitution into
Eq.~\eqref{sm:eq:v2-exact-gue-contraction} yields
\begin{equation}
 \bar q_\infty-\gamma^2
 =\frac{2(d_A-1)(d_Ad_{\bar A}^2+d_{\bar A}+2)}
 {d_Ad_{\bar A}(d_Ad_{\bar A}-1)(d_Ad_{\bar A}+1)(d_Ad_{\bar A}+3)}
 =\mathcal O(d^{-2}),
 \label{sm:eq:v2-gue-late-contraction}
\end{equation}
where the last estimate is for $d\to\infty$ at fixed $d_A$.
Since $\gamma^2=1/d_A+\mathcal O(d^{-2})$, the long-time
retention factor also satisfies
$\bar q_\infty=1/d_A+\mathcal O(d^{-2})$.

To prove the theorem, we now bound the finite-time correction uniformly
for $t\geq\sqrt d$.  Since $|s_1(t)|\leq d$ and $|s_2(t)|\leq d$, the
form factors obey
$\mathcal R_4(t)\leq d^2\mathcal R_2(t)$,
$|\mathcal R_{4,1}(t)|\leq d\mathcal R_2(t)$, and
$\mathcal R_{4,2}(t)\leq d^2$.  Substituting these estimates into
Eq.~\eqref{sm:eq:v2-exact-gue-contraction}, and using
$\gamma^2-1/d_A=-(d_A-1/d_A)/(d^2-1)$, gives
\begin{equation}
 \left|\bar q_t-\frac{1}{d_A}\right|
 \leq C_{d_A}\left[\frac{\mathcal R_2(t)}{d^2}+\frac{1}{d^2}\right],
 \label{sm:eq:v2-gue-spectral-error-bound}
\end{equation}
where $C_{d_A}>0$ is a constant depending only on the fixed reset dimension $d_A$, and the
bound holds for all $d>3$ and all $t$.
It therefore suffices to show that $\mathcal R_2(t)/d^2$ tends to zero
uniformly for $t\geq\sqrt d$ as $d_{\bar A}\to\infty$.

We first bound the fluctuations of $\Tr U(t)$ using the standard GUE
eigenvalue correlation formula in Eq.~(2.3) of
Ref.~\cite{soshnikov2000determinantal}.
Let $h_{\ell,d}(E)$, $\ell=0,\ldots,d-1$, be the real orthonormal
Hermite functions for our GUE normalization: each is a polynomial of
degree $\ell$ multiplied by $\e^{-dE^2/4}$, with
$\int_{\mathbb R}h_{\ell,d}(E)h_{m,d}(E)\,\mathrm dE=\delta_{\ell m}$.
Using the one- and two-eigenvalue correlations to evaluate the variance
of $\Tr U(t)=\sum_a\e^{-\im E_at}$ gives
\begin{equation}
 \mathbb E_H\left|\Tr U(t)-\mathbb E_H\Tr U(t)\right|^2
 =d-\sum_{\ell,m=0}^{d-1}
 \left|\int_{\mathbb R}\e^{-\im Et}
 h_{\ell,d}(E)h_{m,d}(E)\,\mathrm dE\right|^2
 \leq d.
 \label{sm:eq:v2-gue-trace-fluctuation}
\end{equation}
Here equal eigenvalue indices contribute $d$, while the correlations
between distinct eigenvalues give the subtracted squares after removing
the squared mean.  This bound holds for every $d$ and $t$.
Using
$\mathcal R_2(t)=|\mathbb E_H\Tr U(t)|^2
+\mathbb E_H|\Tr U(t)-\mathbb E_H\Tr U(t)|^2$
therefore gives
\begin{equation}
 \frac{\mathcal R_2(t)}{d^2}
 \leq\left|\frac{\mathbb E_H\Tr U(t)}{d}\right|^2+\frac{1}{d}.
 \label{sm:eq:v2-gue-form-factor-bound}
\end{equation}

We next bound $d^{-1}\mathbb E_H\Tr U(t)$ through the mean eigenvalue
density.  Let $p_d(E)$ be the probability density obtained by drawing a
GUE Hamiltonian and then choosing one of its $d$ eigenvalues with equal
probability, so that $\int_{\mathbb R}p_d(E)\,\mathrm dE=1$.
It follows that
\begin{equation}
 \frac{\mathbb E_H\Tr U(t)}d
 =\int_{\mathbb R}\e^{-\im Et}p_d(E)\,\mathrm dE.
 \label{sm:eq:v2-gue-mean-trace-density}
\end{equation}
For the remaining estimates, keep $d_A$ fixed and take
$d_{\bar A}\to\infty$, so that $d=d_Ad_{\bar A}\to\infty$.
The bulk density asymptotics derived in Ref.~\cite{sup_garoni2005density}
imply that, for every fixed $|E|<2$, $p_d(E)$ converges to the semicircle density
$p_{\rm sc}(E)=\sqrt{4-E^2}/(2\pi)$ on $[-2,2]$, with
$p_{\rm sc}(E)=0$ outside this interval.
Here the spectral variable $x$ used in that reference is rescaled by
$E=2x$ to match our GUE normalization, and its normalized density is
therefore evaluated at $E/2$ and multiplied by $1/2$.
Both densities are nonnegative and have unit integral.  Consequently,
\begin{equation}
 \int_{\mathbb R}|p_d(E)-p_{\rm sc}(E)|\,\mathrm dE
 =2-2\int_{\mathbb R}\min\{p_d(E),p_{\rm sc}(E)\}\,\mathrm dE
 \xrightarrow[d\to\infty]{}0.
 \label{sm:eq:v2-gue-density-lone}
\end{equation}
Indeed, the integrand in the second integral is bounded by $p_{\rm sc}$
and converges to it inside $(-2,2)$, while it is zero outside $[-2,2]$.
Dominated convergence therefore makes that integral tend to one.

Subtracting and adding the Fourier transform of $p_{\rm sc}$ in
Eq.~\eqref{sm:eq:v2-gue-mean-trace-density}, and using
$|\e^{-\im Et}|=1$, yields
\begin{align}
 \sup_{t\geq\sqrt d}\left|\frac{\mathbb E_H\Tr U(t)}{d}\right|
 &\leq\int_{\mathbb R}|p_d(E)-p_{\rm sc}(E)|\,\mathrm dE
 \nonumber\\
 &\quad+\sup_{t\geq\sqrt d}
 \left|\int_{\mathbb R}\e^{-\im Et}p_{\rm sc}(E)\,\mathrm dE\right|
 \xrightarrow[d\to\infty]{}0.
 \label{sm:eq:v2-gue-uniform-trace-bound}
\end{align}
The first term tends to zero by Eq.~\eqref{sm:eq:v2-gue-density-lone}.
The Fourier transform in the second term tends to zero as $t\to\infty$
by the Riemann--Lebesgue lemma, since $p_{\rm sc}$ is integrable.
As the lower limit $\sqrt d$ tends to infinity, its supremum over the
whole interval $t\geq\sqrt d$ also tends to zero.
Equation~\eqref{sm:eq:v2-gue-form-factor-bound} now implies
$\sup_{t\geq\sqrt d}\mathcal R_2(t)/d^2\to0$ as $d\to\infty$.
Combining this with Eq.~\eqref{sm:eq:v2-gue-spectral-error-bound} gives
$\bar q_t=1/d_A+o(1)$ uniformly for $t\geq\sqrt d$ as
$d_{\bar A}\to\infty$ at fixed $d_A$.
For the single reset qubit, $d_A=2$, so this is
$\bar q_t=1/2+o(1)$, as required.

\textit{2. Propagation of the memory difference.}
The GUE average is independent of the nonzero memory difference.  Define
\begin{equation}
 \alpha_t:=\bar q_t+\varepsilon_{\rm EDA}=\frac12+o(1), 
 \label{sm:eq:v2-finite-time-worst-contraction}
\end{equation}
as $d_{\bar A}\to\infty$. 
Under the EDA, $\alpha_t$ bounds the ratio of the mean squared memory
differences after and before one cycle.  Indeed,
Eq.~\eqref{sm:eq:v2-eda} gives
\begin{equation}
 \mathbb E_{\rm joint}\|\Delta_{j+1}\|_2^2
 \leq(\bar q_t+\varepsilon_{\rm EDA})
      \mathbb E_{\rm joint}\|\Delta_j\|_2^2
 =\alpha_t\mathbb E_{\rm joint}\|\Delta_j\|_2^2.
 \label{sm:eq:v2-eda-one-step}
\end{equation}
If the current mean squared difference is zero, all subsequent differences
also vanish and the inequality is immediate.
By stationarity, $\mathbb E_{\rm joint}\|\Delta_0\|_2^2
=2\overline{\mathcal I}_1$.  Iterating the one-step bound and using
Eq.~\eqref{sm:eq:v2-ensemble-influence} therefore gives
\begin{equation}
 \overline{\mathcal I}_\tau
 \leq\overline{\mathcal I}_1\alpha_t^{\tau-1}.
 \label{sm:eq:v2-eda-influence}
\end{equation}
For $\tau=1$, there are no further memory updates, and $\alpha_t^0=1$.

\textit{3. From input influence to performance.}
Apply Lemma~\ref{sm:thm:v2-universal-forgetting} to each fixed $U$ and average
over the reservoir ensemble.  Since $f_\lambda$ is increasing and concave,
Jensen's inequality and Eq.~\eqref{sm:eq:v2-eda-influence} give
\begin{align}
 \mathbb E_{U\sim\nu_t}P_\lambda(U)
 \leq\mathbb E_{U\sim\nu_t}f_\lambda[\mathcal I_\tau(U)]
 \leq f_\lambda(\overline{\mathcal I}_\tau)
 \leq f_\lambda(\overline{\mathcal I}_1\alpha_t^{\tau-1})
 \leq\min\left\{1,
 \frac{2\overline{\mathcal I}_1\alpha_t^{\tau-1}}{\lambda}\right\}.
 \label{sm:eq:v2-eda-performance}
\end{align}
Substituting $\alpha_t=1/2+o(1)$ proves
Eq.~\eqref{sm:eq:v2-gue-main-performance-bound}.
\hfill$\blacksquare$

The exponential bound in Eq.~\eqref{sm:eq:v2-gue-main-performance-bound}
becomes a nontrivial constraint on performance once $\tau$ is large enough
that its right-hand side is below one.
The numerical comparison of paired runs in Sec.~\ref{sm:sec:v2-numerics}
checks the EDA discrepancy directly.

\section{Free endpoint: quadratic no-go and interaction activation}
\label{sm:sec:v2-free}

\subsection{Quadratic no-go}
\label{sm:sec:v2-free-nogo}

We map the $N$-qubit reservoir to Majorana operators
$m_1,\ldots,m_{2N}$ via the Jordan--Wigner transformation.  We use an
ordering in which the reset qubit $A$ is first.  The Majorana operators
satisfy $m_p^\dagger=m_p$ and $\{m_p,m_q\}=2\delta_{pq}I$.  Quadratic
Hamiltonians and observables lie in the real span of $I$ and
$\im m_pm_q$ with $1\leq p<q\leq2N$.  This quadratic span is closed under Heisenberg evolution
because quadratic unitaries act linearly on Majorana
operators~\cite{sup_bravyi2005fermionic}.
Recall that $z(u)=u-1/2$ and $\phi(u)=2\sqrt{u(1-u)}$.

\begin{lemma}[Quadratic current-input normal form]
\label{sm:lem:v2-current-normal-form}
For arbitrary quadratic dynamics and any quadratic feature at time $k$,
\begin{equation}
    x_{k,i}
    =a_{k,i}^{(0)}+\phi(u_k)a_{k,i}^{(1)}-2\zeta_i z_k,
    \label{sm:eq:v2-current-normal-form}
\end{equation}
where $a_{k,i}^{(0)}$ and $a_{k,i}^{(1)}$ are real functions of
the input history preceding $u_k$, whereas $\zeta_i$ is a real coefficient
determined by $U$ and $O_i$, independent of the inputs.
\end{lemma}

\textit{Proof.}
A quadratic Hamiltonian acts linearly on Majoranas, so Heisenberg evolution
preserves the span of $I$ and the bilinears.  Let
$\gamma_1,\ldots,\gamma_{2(N-1)}$ denote the Majorana operators on the memory
subsystem, with indices $r,s=1,\ldots,2(N-1)$.
In the standard Jordan--Wigner convention with the reset mode first,
\begin{equation*}
    m_1=X_A\otimes I_{\bar{A}},\qquad
    m_2=Y_A\otimes I_{\bar{A}},\qquad
    m_{r+2}=Z_A\otimes\gamma_r.
\end{equation*}
The possible bilinears then satisfy
\begin{align*}
    \im m_{r+2}m_{s+2}&=I_A\otimes\im\gamma_r\gamma_s,
    &\im m_1m_2&=-Z_A\otimes I_{\bar{A}},\\
    \im m_1m_{r+2}&=Y_A\otimes\gamma_r,
    &\im m_2m_{r+2}&=-X_A\otimes\gamma_r.
\end{align*}
Since $U^\dagger O_iU$ is a real linear combination of $I$ and these
bilinears, grouping the terms that contain zero, one, or two reset-mode
Majoranas gives
\begin{equation}
    U^\dagger O_iU
    =I_A\otimes Q_{\bar{A},i}+\zeta_i Z_A\otimes I_{\bar{A}}
      +X_A\otimes R_{\bar{A},i}+Y_A\otimes S_{\bar{A},i},
    \label{sm:eq:v2-quadratic-reset-decomposition}
\end{equation}
where, for real coefficients,
\begin{equation}
    Q_{\bar{A},i}=q_{i,0}I_{\bar{A}}+\sum_{r<s}q_{i,rs}\,\im\gamma_r\gamma_s,
    \qquad
    R_{\bar{A},i}=\sum_r r_{i,r}\gamma_r,
    \qquad
    S_{\bar{A},i}=\sum_r s_{i,r}\gamma_r.
    \label{sm:eq:v2-memory-operator-definitions}
\end{equation}
Thus $Q_{\bar{A},i}$ is explicitly a quadratic operator on the memory subsystem,
including a possible identity component, whereas $R_{\bar{A},i}$ and $S_{\bar{A},i}$ are
linear in the memory Majoranas.  In particular, the two-reset sector contains
only $\im m_1m_2$, so the coefficient $\zeta_i$ carries no memory operator.
Taking the expectation in $\rho_{\rm in}(u_k)\otimes\sigma_{k-1}$ and using
\begin{equation}
    \Tr[X\rho_{\rm in}(u_k)]=\phi(u_k),\qquad
    \Tr[Y\rho_{\rm in}(u_k)]=0,\qquad
    \Tr[Z\rho_{\rm in}(u_k)]=-2z_k,
\end{equation}
gives Eq.~\eqref{sm:eq:v2-current-normal-form}, with
$a_{k,i}^{(0)}=\Tr(Q_{\bar{A},i}\sigma_{k-1})$ and
$a_{k,i}^{(1)}=\Tr(R_{\bar{A},i}\sigma_{k-1})$.  The coefficient $\zeta_i$ is
history independent because its memory factor is the identity.
\hfill$\blacksquare$

\begin{theorem}[Cross-time product no-go]
\label{sm:thm:v2-free-nogo}
Under the amplitude encoding in Eq.~\eqref{sm:eq:v2-amplitude-encoding}, let the
inputs be i.i.d. with a distribution symmetric about $1/2$, and assume a
causal stationary reservoir process.  For any quadratic Hamiltonian and any
collection of quadratic observables, let
$h_{k-1}:=h(\ldots,u_{k-2},u_{k-1})$ be a nonzero square-integrable function
of earlier inputs with $\mathbb E_{\rm in} h_{k-1}=0$.  Then
\begin{equation}
    y_k=z_kh_{k-1}
    \quad\Longrightarrow\quad
    P_\lambda=0,
    \qquad
    \operatorname{NMSE}_{\rm pop}^{\rm opt}=1.
    \label{sm:eq:v2-free-nogo}
\end{equation}
Here $\operatorname{NMSE}_{\rm pop}^{\rm opt}$ is the minimum population
NMSE over all affine readouts of the chosen features.
\end{theorem}

\textit{Proof.}
By Lemma~\ref{sm:lem:v2-current-normal-form}, independence of the current input
from the past and reflection symmetry give, term by term,
\begin{align}
    \mathbb E_{\rm in}[z_kh_{k-1}a_{k,i}^{(0)}]
    &=\mathbb E_u[z(u)]\,\mathbb E_{\rm in}[h_{k-1}a_{k,i}^{(0)}]=0,\nonumber\\
    \mathbb E_{\rm in}[z_kh_{k-1}\phi(u_k)a_{k,i}^{(1)}]
    &=\mathbb E_u[z(u)\phi(u)]\,
      \mathbb E_{\rm in}[h_{k-1}a_{k,i}^{(1)}]=0,\nonumber\\
    \mathbb E_{\rm in}[z_k^2h_{k-1}]
    &=\mathbb E_u[z(u)^2]\,\mathbb E_{\rm in}h_{k-1}=0.
    \label{sm:eq:v2-free-orthogonality}
\end{align}
Thus every quadratic feature, and also the intercept, is orthogonal to the
centered target.  Equation~\eqref{sm:eq:v2-exact-ridge-gain} therefore gives
$P_\lambda=0$.  Any affine prediction
$\bar y_k=b+\bm w^\top\bm x_k$ also obeys
\begin{equation}
    \mathbb E_{\rm in}[(y_k-\bar y_k)^2]
    =\mathbb E_{\rm in}[y_k^2]+\mathbb E_{\rm in}[\bar y_k^2]
    \geq\mathbb E_{\rm in}[y_k^2].
\end{equation}
The constant mean predictor attains the lower bound.  Since
$\sigma_y^2=\mathbb E_u[z(u)^2]\mathbb E_{\rm in}[h_{k-1}^2]>0$,
the optimal population NMSE is one.
\hfill$\blacksquare$

Although the theorem is stated for centered cross-time targets, the same
inaccessible component also limits uncentered targets when it is orthogonal
to the remainder.  The obstruction still depends on the input encoding.  The
delayed product provides a direct example: even perfect recovery of its
additive part cannot remove the inaccessible centered cross-time component.

\begin{smcorollary}[Delayed-product performance bound]
\label{sm:cor:v2-raw-product-floor}
For uniform i.i.d. inputs and a causal stationary process, arbitrary
quadratic dynamics and any collection of quadratic observables obey, for
the target $y_k=u_ku_{k-\tau}$, every $\tau\geq1$, and every $\lambda>0$,
\begin{equation}
 P_\lambda\leq\frac67.
 \label{sm:eq:v2-raw-product-floor}
\end{equation}
\end{smcorollary}

\textit{Proof.}
Put $s=u_k-1/2$, $r=u_{k-\tau}-1/2$, and center the target:
\begin{equation}
 u_ku_{k-\tau}-\mathbb E_{\rm in}[u_ku_{k-\tau}]
 =y_\times+y_{\rm add},\qquad y_\times=sr,\qquad
 y_{\rm add}=\frac{s+r}{2}.
\end{equation}
For fixed readout weights $\bm w$, the mean-squared error is minimized
over the intercept by
$b^*(\bm w)=\mathbb E_{\rm in}y_k
-\bm w^\top\mathbb E_{\rm in}\bm x_k$.
Write $g_k=\bm w^\top(\bm x_k-\mathbb E_{\rm in}\bm x_k)$ for the
corresponding centered prediction.
Theorem~\ref{sm:thm:v2-free-nogo} gives
$\mathbb E_{\rm in}[y_\times g_k]=0$, while independence and centering
give $\mathbb E_{\rm in}[y_\times y_{\rm add}]=0$.
The minimum population prediction error over all affine readouts is therefore
\begin{align}
 \operatorname{MSE}_{\rm pop}^{\rm opt}
 &:=\min_{b,\bm w}
 \mathbb E_{\rm in}[(y_k-b-\bm w^\top\bm x_k)^2]\nonumber\\
 &=\min_{\bm w}\mathbb E_{\rm in}[(y_\times+y_{\rm add}-g_k)^2]\nonumber\\
 &=\mathbb E_{\rm in}[y_\times^2]
   +\min_{\bm w}\mathbb E_{\rm in}[(y_{\rm add}-g_k)^2]
 \geq\frac1{144},
\end{align}
where $\mathbb E_{\rm in}[y_\times^2]
=\mathbb E_{\rm in}[s^2]\mathbb E_{\rm in}[r^2]=1/144$.
Since $\operatorname{Var}(y_{\rm add})=1/24$ and
$\operatorname{Var}(u_ku_{k-\tau})=1/144+1/24=7/144$,
any affine prediction has normalized error at least $1/7$.  A ridge readout is
affine, so Eq.~\eqref{sm:eq:v2-raw-product-floor} follows.
\hfill$\blacksquare$

\subsection{Interaction activation under the rewinding protocol}
\label{sm:sec:v2-activation}

We next study how quartic interactions lift the free-endpoint obstruction
under the QRC rewinding protocol~\cite{mujal2023measurements}.  For each readout,
the reservoir is prepared in the same fixed initial state and processes a
fixed-length window of inputs ending at $u_k$, repeating reset and unitary
evolution for each input.  The observables are measured after the final
update, and the expectations below are over the inputs in this window.
A pure-quartic operator is a real linear combination of products of four
distinct Majorana operators.

We use the interpolation
\begin{equation}
    H_\theta=\cos\theta\,H_2+\sin\theta\,H_4,
    \label{sm:eq:v2-syk-interpolation}
\end{equation}
where $\theta\in\mathbb R$, $H_2$ is quadratic, and $H_4$ is pure quartic.
Each update uses the same evolution time $t>0$.
Let $\rho_k(\theta)$ be the final state and let $\bm x_k(\theta)$ contain
$x_{k,i}(\theta):=\Tr[O_i\rho_k(\theta)]$ for quadratic observables
$\{O_i\}$, taken to be Hilbert--Schmidt orthonormal.
For the target $y_k=z_kz_{k-\tau}$, define
$\bm c(\theta):=\operatorname{Cov}[\bm x_k(\theta),y_k]$,
and let $P_\lambda(\theta)$ be the corresponding population performance.
A dot denotes differentiation with respect to $\theta$, so
$\dot{\bm c}(0):=\left.\partial_\theta\bm c(\theta)\right|_{\theta=0}$.

\begin{proposition}[Interaction activation under the rewinding protocol]
\label{sm:prop:v2-activation}
Let $N\geq3$ and $1\leq\tau\leq N-1$.  Use amplitude encoding with
i.i.d. inputs whose distribution is symmetric about $1/2$ and has nonzero
variance.  For the target $y_k=z_kz_{k-\tau}$, use a fixed rewinding window
of at least $\tau+1$ updates and a fixed initial state independent of
the inputs and $\theta$.  Take the quadratic readout set $\{O_i\}$ to be
fixed and complete.  Then there exist $H_2$, $H_4$, and $t>0$ for which
\begin{equation}
    \bm v_\tau:=\dot{\bm c}(0)\neq0.
    \label{sm:eq:v2-activation-response}
\end{equation}
Whenever $\bm v_\tau\neq0$, every fixed $\lambda>0$ gives
\begin{equation}
    \bm c(\theta)=\theta\bm v_\tau+O(\theta^2),
    \qquad
    P_\lambda(\theta)=\kappa_\lambda\theta^2+O(\theta^3),
    \quad \kappa_\lambda>0.
    \label{sm:eq:v2-activation-gain}
\end{equation}
\end{proposition}

This nonzero response also holds almost surely for independent full-support
Gaussian quadratic and quartic couplings at every fixed evolution time
outside a discrete exceptional set.  The statement remains valid after
separately normalizing $H_2$ and $H_4$ by
$d^{-1}\Tr(H_p^2)=1$ for $p=2,4$.
We prove these claims together below.

\textit{Proof.}
Label the reset mode by $a=1$ and fix $t>0$.  Let $\sigma_{k-1}(\theta)$ be
the memory state in this run immediately before processing $u_k$ under
$H_\theta$.  Choose a fermionic Gaussian unitary
that cyclically permutes the $N$ modes, and choose its static quadratic
generator $H_2^{\rm reg}$ so that $\e^{-\im tH_2^{\rm reg}}$ implements this
permutation.  Because $\tau\leq N-1$, it stores the input from step $k-\tau$
in a memory mode $b=b_\tau$ immediately before the final evolution.  Let
$m_{j,1}:=m_{2j-1}$ and $m_{j,2}:=m_{2j}$ denote the two Majoranas
of mode $j=1,\ldots,N$, and set $Z_j:=-\im m_{j,1}m_{j,2}$.
This is the Pauli-$Z$ operator on mode $j$ in the convention used above.  Define
$\Gamma_\tau(H_2):=\operatorname{Cov}(\langle Z_b\rangle_0,z_{k-\tau})$, where
$\langle Z_b\rangle_0:=\Tr_{\bar{A}}(Z_b\sigma_{k-1}(0))$.
The subscript $0$ indicates $\theta=0$.  At the exact-register
point, $\langle Z_b\rangle_0=-2z_{k-\tau}$ and hence
\begin{equation}
    \Gamma_\tau(H_2^{\rm reg})
    =-2\mathbb E_u[z(u)^2]\neq0,
    \label{sm:eq:v2-register-memory}
\end{equation}
and continuity gives $\Gamma_\tau(H_2)\neq0$ throughout a neighborhood of
$H_2^{\rm reg}$.

Having retained information about $z_{k-\tau}$, we now construct a
quartic interaction that couples it to the current input.
We want the first-order change of a quadratic readout to contain $Z_aZ_b$:
the reset factor $Z_a$ supplies $z_k$, and $Z_b$ carries information about
the delayed input.  To choose a static $H_4$ that produces this term,
we must account for its action throughout the evolution interval.
For a candidate Hermitian pure-quartic interaction $D_4$, define
\begin{equation}
    \mathcal J_{H_2}(D_4)
    :=\int_0^t\e^{\im sH_2}D_4\e^{-\im sH_2}\,\mathrm ds.
    \label{sm:eq:v2-duhamel-map}
\end{equation}
The integrand is the interaction in the interaction picture of $H_2$,
and the integral gives its accumulated first-order effect over time $t$.
To see how this enters the readout, write
$U_\theta:=\e^{-\im tH_\theta}$ and $U_0:=\e^{-\im tH_2}$.
Since $\partial_\theta H_\theta|_{\theta=0}=H_4$,
Duhamel's formula~\cite{sup_almohy2009frechet} gives
$U_\theta=U_0[I-\im\theta\mathcal J_{H_2}(H_4)]+O(\theta^2)$.
For a quadratic operator $Q_0$ acting only on $\bar A$, choose the quadratic
readout $Q=U_0Q_0U_0^\dagger$ and keep it fixed as $\theta$ varies.
Its unperturbed Heisenberg evolution is $U_0^\dagger Q U_0=Q_0$, and
\begin{equation}
 \left.\partial_\theta(U_\theta^\dagger Q U_\theta)\right|_{\theta=0}
 =\im[\mathcal J_{H_2}(H_4),Q_0].
 \label{sm:eq:v2-duhamel-readout-response}
\end{equation}
Thus we need to choose $Q_0$ and $H_4$ so that this commutator equals
$Z_aZ_b$.  We will find $H_4$ by inverting $\mathcal J_{H_2}$.

Quadratic conjugation preserves the finite-dimensional space of
pure-quartic operators, so $\mathcal J_{H_2}$ is a real-linear map from
this space to itself.
Its determinant is real analytic in the quadratic couplings and is not
identically zero because $\mathcal J_0(D_4)=tD_4$.
The set where the determinant is nonzero is therefore dense.
We may consequently choose $H_2$ within the neighborhood above so that
both $\Gamma_\tau(H_2)\neq0$ and $\mathcal J_{H_2}$ is invertible.
Choose a Majorana $m_r$ in a third mode, distinct from $a$ and $b$, and
set the quadratic memory operator $Q_0=\im m_{b,2}m_r$.
The Hermitian pure-quartic operator
\begin{equation}
    G_\tau=-\frac12m_{a,1}m_{a,2}m_{b,1}m_r
\end{equation}
satisfies, by the Clifford relations,
\begin{equation}
    \im[G_\tau,Q_0]=Z_aZ_b.
\end{equation}
Taking $H_4=\mathcal J_{H_2}^{-1}(G_\tau)$ now realizes the required
integrated interaction with a static pure-quartic Hamiltonian.
For the fixed readout $Q=U_0Q_0U_0^\dagger$, define
$x_Q(\theta):=\Tr[Q\rho_k(\theta)]$.
Equation~\eqref{sm:eq:v2-duhamel-readout-response} then gives
\begin{equation}
    \left.\partial_\theta(U_\theta^\dagger Q U_\theta)
    \right|_{\theta=0}=Z_aZ_b.
    \label{sm:eq:v2-activation-source}
\end{equation}

The earlier updates also change the memory state at first order.
We now show that their contribution cannot cancel the target response
created in the final update.
Writing $\dot\sigma_{k-1}:=\partial_\theta\sigma_{k-1}(\theta)|_{\theta=0}$ and
differentiating the final-step expectation gives
\begin{equation}
    \left.\partial_\theta x_Q(\theta)\right|_{\theta=0}
    =-2z_k\langle Z_b\rangle_0+\Tr_{\bar{A}}(Q_0\dot\sigma_{k-1}).
    \label{sm:eq:v2-activation-feature-derivative}
\end{equation}
The first term couples the current input to the information stored in
mode $b$.  The second term contains the effect of all earlier updates.
Because $Q_0$ acts only on the memory subsystem, this term depends only on
inputs preceding $u_k$.  Its covariance with $z_kz_{k-\tau}$ is therefore
zero, since $z_k$ is centered and independent of the past.
Because the target also has zero mean, Eqs.~\eqref{sm:eq:v2-register-memory}
and \eqref{sm:eq:v2-activation-feature-derivative} give
\begin{equation}
    \left.\partial_\theta
      \operatorname{Cov}(x_Q(\theta),z_kz_{k-\tau})\right|_{\theta=0}
    =-2\mathbb E_u[z(u)^2]
      \Gamma_\tau(H_2)\neq0.
\end{equation}

The observable $Q$ is a linear combination of the measured quadratic
observables.  Its nonzero covariance response therefore implies
$\bm v_\tau\neq0$.  Regard
$F_\tau:=\|\bm v_\tau\|^2$ as a function of $H_2$, $H_4$, and $t$.
For a fixed rewinding window, $F_\tau$ is real analytic in these variables, and the
witness above shows that it is not identically zero.  Holding the witness
couplings fixed, its zeros as a function of $t$ are discrete.  At any other
fixed time, $F_\tau$ is a nonzero real-analytic function of the couplings, so
its zero set has Lebesgue measure zero.  This proves the almost-sure claim for
full-support Gaussian couplings~\cite{sup_mityagin2020zero}.

For separate normalization, write
$\mathcal N_2=[d^{-1}\Tr(H_2^2)]^{1/2}$ and
$\mathcal N_4=[d^{-1}\Tr(H_4^2)]^{1/2}$.  Replacing
$(H_2,H_4,t)$ by $(H_2/\mathcal N_2,H_4/\mathcal N_4,\mathcal N_2t)$
leaves every unperturbed update unchanged and multiplies its full first-order
response by $\mathcal N_2/\mathcal N_4\neq0$.  Hence there is
also a nonzero witness on the product of normalized coupling spheres.  The
same analytic-zero-set argument, in local sphere coordinates, gives surface
measure zero there; a normalized full-support Gaussian has an absolutely
continuous directional law.  Enlarging the discrete exceptional set if
necessary proves the normalized claim.

Reflection symmetry gives $\mathbb E_u[z(u)]=\mathbb E_u[z(u)\phi(u)]=0$.  The termwise
orthogonality argument of Theorem~\ref{sm:thm:v2-free-nogo} uses only
independence of the current input from the past and therefore also gives
$\bm c(0)=0$ for the rewinding protocol with its input-independent
initial state.
Let $\Sigma(\theta):=
\operatorname{Cov}[\bm x_k(\theta),\bm x_k(\theta)]$ and
$K_\lambda(\theta):=[\Sigma(\theta)+\lambda I]^{-1}$.
The fixed rewinding window makes both covariance functions real analytic, and $\lambda>0$
ensures that $K_\lambda(\theta)$ is analytic near zero.  Hence
$\bm c(\theta)=\theta\bm v_\tau+O(\theta^2)$ and
$K_\lambda(\theta)=K_\lambda(0)+O(\theta)$.  Substitution into
Eq.~\eqref{sm:eq:v2-exact-ridge-gain} gives Eq.~\eqref{sm:eq:v2-activation-gain}, with
\begin{equation}
    \kappa_\lambda=\frac{\bm v_\tau^\top
      [K_\lambda(0)+\lambda K_\lambda(0)^2]\bm v_\tau}
      {\sigma_y^2}.
\end{equation}
The bracketed matrix is positive definite, so $\kappa_\lambda>0$ because
$\bm v_\tau\neq0$.
\hfill$\blacksquare$

\section{Exact stationary population performance across the crossover}
\label{sm:sec:v2-two-copy}

To compute the stationary population performance in main-text Fig.~2, we
need the feature--target covariance vector and the feature covariance
$\Sigma$.  Lemma~\ref{sm:lem:v2-centered-product-covariance} gives the
former, and Theorem~\ref{sm:thm:v2-two-copy} gives the latter.  We then
substitute both into the population performance formula
\eqref{sm:eq:v2-exact-ridge-gain}.

For uniform i.i.d. inputs, consider
\begin{equation}
    y_k=z_kz_{k-\tau},
    \qquad z(u)=u-\frac12,
    \qquad \tau\geq1.
    \label{sm:eq:v2-product-task}
\end{equation}
For any operator $Q$ on $A$ and operator $D$ on $\bar A$, define
\begin{equation}
    \Phi_{Q,U}(D):=\Tr_A[U(Q\otimes D)U^\dagger].
    \label{sm:eq:v2-component-map}
\end{equation}
This linear map describes how the component $Q\otimes D$ contributes to
the memory state after one evolution and partial trace over $A$.  In
particular, $\Phi_{\rho_{\rm in}(u),U}=\Phi_{u,U}$ is the memory update
for input $u$.  The mean input state, the memory map averaged over the input,
and the stationary mean memory state are, respectively,
\begin{equation}
    \overline\rho_{\rm in}:=\mathbb E_u\rho_{\rm in}(u)
      =\frac I2+\frac{\pi}{8}X,
    \qquad
    \overline\Phi_U:=\mathbb E_u\Phi_{u,U}=\Phi_{\overline\rho_{\rm in},U},
    \qquad
    \overline\sigma:=\mathbb E_{\rm in}\sigma_k
      =\overline\Phi_U(\overline\sigma).
    \label{sm:eq:v2-mean-channel}
\end{equation}
To track the memory component correlated with the old input $u_{k-\tau}$,
define
\begin{equation}
    D_{\tau,Z}
    :=\overline\Phi_U^{\,\tau-1}
      [\Phi_{Z,U}(\overline\sigma)].
    \label{sm:eq:v2-D-response}
\end{equation}
The $Z$ component carries the centered input into the memory, and
$\overline\Phi_U^{\,\tau-1}$ propagates it through the next $\tau-1$
cycles after averaging over their inputs.  For $\tau=1$, no further
memory update is applied.  As shown below,
$D_{\tau,Z}=-12\,\mathbb E_{\rm in}[z_{k-\tau}\sigma_{k-1}]$.
For this target, write
$\bm c_\tau(U):=\operatorname{Cov}(\bm x_k,y_k)$, with components
$c_{\tau,i}(U):=\operatorname{Cov}(x_{k,i},y_k)$, keeping $U$ fixed.
This is the vector $\bm c$ in Sec.~\ref{sm:sec:v2-protocol}.

\begin{lemma}[Exact centered-product covariance]
\label{sm:lem:v2-centered-product-covariance}
For the stationary causal process above, the fixed-$U$ feature--target
covariance for Eq.~\eqref{sm:eq:v2-product-task} is
\begin{equation}
    c_{\tau,i}(U)
    =\frac1{144}
      \Tr[O_iU(Z\otimes D_{\tau,Z})U^\dagger].
    \label{sm:eq:v2-exact-covariance-direction}
\end{equation}
\end{lemma}

\textit{Proof.}
The required input moments are
\begin{equation}
    \mathbb E_u[z(u)]=0,\qquad
    \mathbb E_u[z(u)^2]=\frac1{12},\qquad
    \mathbb E_u[z(u)\rho_{\rm in}(u)]=-\frac Z{12}.
    \label{sm:eq:v2-z-moments}
\end{equation}
Let $j=k-\tau$.  Since $u_j$ is independent of the preceding memory state,
\begin{equation}
    \mathbb E_{\rm in}[z_j\sigma_j]
    =-\frac1{12}\Phi_{Z,U}(\overline\sigma).
\end{equation}
Averaging the next $\tau-1$ independent updates therefore gives
\begin{equation}
    \mathbb E_{\rm in}[z_{k-\tau}\sigma_{k-1}]
    =-\frac1{12}D_{\tau,Z}.
    \label{sm:eq:v2-lag-score-propagation}
\end{equation}
At the final reset, the independent current input contributes a second factor
$-Z/12$.  Because $\mathbb E_{\rm in}y_k=0$, the target--state covariance
operator is
\begin{equation}
    C_\tau(U):=\mathbb E_{\rm in}[y_k\rho_k]
    =\frac1{144}U(Z\otimes D_{\tau,Z})U^\dagger.
    \label{sm:eq:v2-exact-covariance-operator}
\end{equation}
Contracting with $O_i$ proves
Eq.~\eqref{sm:eq:v2-exact-covariance-direction}.
\hfill$\blacksquare$

The feature covariance is contained in the stationary second moment
\begin{equation}
    \Omega_U^{(2)}:=\mathbb E_{\rm in}[\sigma_k\otimes\sigma_k].
    \label{sm:eq:v2-two-copy-moment}
\end{equation}
The superscript $(2)$ denotes two copies, so $\Omega_U^{(2)}$ acts on
$\mathcal H_{\bar A}\otimes\mathcal H_{\bar A}$.
Both factors represent the same random memory state generated by the same
input history.  Their product allows us to calculate products of feature
values from that history.  For a memory operator $D$, define the full-state update
\begin{equation}
    \mathcal E_{u,U}(D):=U[\rho_{\rm in}(u)\otimes D]U^\dagger.
    \label{sm:eq:v2-full-state-update}
\end{equation}
For a memory state $D$, $\mathcal E_{u,U}(D)$ is the full reservoir
state used for readout, whereas $\Phi_{u,U}(D)$ is the memory state
obtained by tracing out $A$.
Write $\mu_{x,i}:=\mathbb E_{\rm in}x_{k,i}$ for the mean of feature $i$.

\begin{theorem}[Exact two-copy feature covariance]
\label{sm:thm:v2-two-copy}
Assume i.i.d. inputs and a stationary causal process
$\sigma_k=\Phi_{u_k,U}(\sigma_{k-1})$.  Then
\begin{equation}
    \Omega_U^{(2)}
    =\mathbb E_u[(\Phi_{u,U}\otimes\Phi_{u,U})(\Omega_U^{(2)})].
    \label{sm:eq:v2-two-copy-fixed-point}
\end{equation}
With $\overline\sigma=\mathbb E_{\rm in}\sigma_k$, the stationary feature mean and
covariance are
\begin{align}
    \mu_{x,i}
    &=\mathbb E_u\Tr[O_i\mathcal E_{u,U}(\overline\sigma)],\nonumber\\
    (\Sigma)_{ij}
    &=\mathbb E_u\Tr\!\left[
      (O_i\otimes O_j)
      (\mathcal E_{u,U}\otimes\mathcal E_{u,U})(\Omega_U^{(2)})
      \right]-\mu_{x,i}\mu_{x,j}.
    \label{sm:eq:v2-two-copy-feature-covariance}
\end{align}
Thus $\Sigma$ is the exact fixed-$U$ stationary population covariance.
\end{theorem}

\textit{Proof.}
Independence of the current input and $\sigma_{k-1}$, together with
stationarity, gives
\begin{equation}
    \mathbb E_{\rm in}[\sigma_k\otimes\sigma_k]
    =\mathbb E_u\!\left[
      (\Phi_{u,U}\otimes\Phi_{u,U})
      \mathbb E_{\rm in}(\sigma_{k-1}\otimes\sigma_{k-1})\right],
\end{equation}
which proves Eq.~\eqref{sm:eq:v2-two-copy-fixed-point}.  Applying the same current
input to both copies of $\sigma_{k-1}$ gives
\begin{equation*}
 \mathbb E_{\rm in}[x_{k,i}x_{k,j}]
 =\mathbb E_u\Tr\!\left[(O_i\otimes O_j)
 (\mathcal E_{u,U}\otimes\mathcal E_{u,U})(\Omega_U^{(2)})\right].
\end{equation*}
The corresponding one-copy average is $\mu_{x,i}$ in
Eq.~\eqref{sm:eq:v2-two-copy-feature-covariance}; subtracting
$\mu_{x,i}\mu_{x,j}$ proves the covariance formula.
\hfill$\blacksquare$

Accordingly, the two-copy channel is
$\mathbb E_u(\Phi_{u,U}\otimes\Phi_{u,U})$;
$\overline\Phi_U\otimes\overline\Phi_U$ instead describes independently
driven copies and omits their input-induced correlations.  In finite dimension,
Eq.~\eqref{sm:eq:v2-two-copy-fixed-point} is linear, and the stationary causal
process determines its physical solution.  Numerically, we solve it in traceless
Hermitian coordinates, removing the fixed trace coordinate without changing
the population quantity.

For the uniform amplitude encoding,
\begin{equation}
    \rho_{\rm in}(u)=\overline\rho_{\rm in}
      +\frac{\phi(u)-\pi/4}{2}X-z(u)Z,
\end{equation}
where
\begin{equation}
    \mathbb E_u[(\phi(u)-\pi/4)^2]=\frac23-\frac{\pi^2}{16},
    \qquad
    \mathbb E_u[z(u)^2]=\frac1{12},
    \qquad
    \mathbb E_u[(\phi(u)-\pi/4)z(u)]=0.
\end{equation}
Linearity reduces the two-copy mean map to
\begin{align}
    \mathbb E_u[\Phi_{u,U}\otimes\Phi_{u,U}]
    &=\overline\Phi_U\otimes\overline\Phi_U
      +\left(\frac16-\frac{\pi^2}{64}\right)
       \Phi_{X,U}\otimes\Phi_{X,U}
      +\frac1{12}\Phi_{Z,U}\otimes\Phi_{Z,U}.
    \label{sm:eq:v2-uniform-two-copy}
\end{align}

Independence of the current and delayed inputs gives
$\sigma_y^2=(\mathbb E_u[z(u)^2])^2=1/144$.  Combining
Eqs.~\eqref{sm:eq:v2-exact-covariance-direction} and
\eqref{sm:eq:v2-two-copy-feature-covariance} with
Eq.~\eqref{sm:eq:v2-exact-ridge-gain} therefore gives the fixed-$U$ stationary
population performance
\begin{equation}
    P_\lambda(U)
    =144\,\bm c_\tau(U)^\top
      (K_\lambda+\lambda K_\lambda^2)\bm c_\tau(U),
    \qquad
    K_\lambda=(\Sigma+\lambda I)^{-1}.
    \label{sm:eq:v2-two-copy-population-gain}
\end{equation}
Equation~\eqref{sm:eq:v2-two-copy-population-gain} is evaluated by solving the
stationary equations for each fixed $U$; the reported ensemble intervals
quantify finite-realization sampling uncertainty.

\section{Additional numerical results}
\label{sm:sec:v2-numerics}

The preceding analytic results are supplemented by numerical finite-size tests
of the model-specific EDA signature and influence recursion, the free no-go
and its linear-memory control, and stationary population performance across the
quadratic--quartic crossover.
All calculations use i.i.d. uniform inputs on $[0,1]$, the amplitude encoding
of Eq.~\eqref{sm:eq:v2-amplitude-encoding}, and ideal expectation values without
measurement shot noise.
In this section, numerical estimates use the same symbols as their theoretical
counterparts; the sampling procedures and uncertainty estimates are specified
below.  For finite-sample tests, $\mathrm{NMSE}_{\rm test}$ is the test
mean-squared error divided by the empirical target variance on that test block.
We estimate $P_\lambda$ by $1-\mathrm{NMSE}_{\rm test}$.  These finite-sample
estimates can be negative.  The crossover calculation instead evaluates the
stationary population performance from the fixed-point moments in
Sec.~\ref{sm:sec:v2-two-copy}.

\subsection{GUE forgetting}

The GUE Hamiltonians follow Eq.~\eqref{sm:eq:v2-normalized-gue} without
realization-wise rescaling.  The evolution times for the finite-size test
in main-text Fig.~1(a) are listed in Table~\ref{sm:tab:v2-eda-scaling} and
satisfy $t\geq\sqrt d$ for every system size.

At these times, we use uniform i.i.d. amplitude inputs, 64 independently
drawn Hamiltonians at each size, and 16 input histories per Hamiltonian.  A
500-step washout prepares the common memory state $\sigma$.  The two runs
then receive inputs $v$ and $v'$, drawn independently from the same input
distribution, followed by identical inputs:
\begin{equation}
 \Delta_0=\Phi_{v,U(t)}(\sigma)-\Phi_{v',U(t)}(\sigma),\qquad
 \Delta_{\tau-1}=\Phi_{u_{\tau-1},U(t)}\circ\cdots\circ
 \Phi_{u_1,U(t)}(\Delta_0).
\end{equation}
Here $u_1,\ldots,u_{\tau-1}$ are the shared inputs following $v$ and $v'$,
and the same $U(t)$ is used throughout both runs.
For $\tau=1$, the composition is empty and $\Delta_{\tau-1}=\Delta_0$.
Because the terminal
reset state is pure and the final evolution is unitary,
$\frac12\|\Delta_{\tau-1}\|_2^2$ equals the full-state lag-influence sample at the
observation time.

For every nonzero reachable $\Delta_{j-1}$ we compute the next-input integral in closed form,
\begin{equation}
 q_{U(t)}(\Delta_{j-1})=\int_0^1
 \frac{\|\Phi_{u,U(t)}(\Delta_{j-1})\|_2^2}{\|\Delta_{j-1}\|_2^2}\,\mathrm du,
\end{equation}
and form the influence-weighted coefficient
using the joint average over Hamiltonian realizations and paired input
histories,
\begin{equation}
 \begin{aligned}
 R_j&:=\frac{\mathbb E_{\rm joint}\|\Delta_j\|_2^2}
 {\mathbb E_{\rm joint}\|\Delta_{j-1}\|_2^2}\\
 &=\frac{\mathbb E_{\rm joint}[\|\Delta_{j-1}\|_2^2q_{U(t)}(\Delta_{j-1})]}
 {\mathbb E_{\rm joint}\|\Delta_{j-1}\|_2^2}
 =\frac{\overline{\mathcal I}_{j+1}}
 {\overline{\mathcal I}_j}.
 \end{aligned}
 \label{sm:eq:v2-numeric-eda-statistic}
\end{equation}
The equalities follow from Eq.~\eqref{sm:eq:v2-exact-influence-recursion}.
We estimate $R_j$ from the middle expression by averaging its numerator
and denominator separately over the sampled histories and Hamiltonians,
then taking their ratio.  Denote this estimate for $N$ qubits by
$R_{j,N}$.
For the GUE and every nonzero $\Delta$, $\bar q_t(\Delta)=\bar q_t$.
Averaging the fixed-spectrum expression in
Eq.~\eqref{sm:eq:v2-gue-contraction-fixed-spectrum} over the same 64 sampled
spectra gives $\bar q_{t,N}$, an estimate of $\bar q_t$ for
$N$ qubits.  The EDA discrepancy estimated over the first seven common
updates is therefore
\begin{equation}
 \varepsilon_{{\rm EDA},N}^{(7)}
 :=\max_{1\leq j\leq7}
 \left|R_{j,N}-\bar q_{t,N}\right|.
 \label{sm:eq:v2-numeric-eda-estimator}
\end{equation}
For each of $10^4$ bootstrap replicates~\cite{sup_efron1979bootstrap}, we
resample the Hamiltonian realizations and use the same indices for every
$R_{j,N}$ and for $\bar q_{t,N}$.
The main-panel bands in main-text Fig.~1(a) are pointwise $95\%$ percentile
intervals.  For the inset, a centered maximum-norm interval accounts
simultaneously for the maximization over the seven updates at each size.

\begin{table}[h!]
\centering
\caption{Finite-size EDA results.  The evolution times $t$ satisfy
$t\geq\sqrt d$.  The last column gives
$\alpha_{t,N}=\bar q_{t,N}+
\varepsilon_{{\rm EDA},N}^{(7)}$, the corresponding estimate of
the influence contraction factor $\alpha_t$.}
\label{sm:tab:v2-eda-scaling}
\begin{ruledtabular}
\begin{tabular}{ccccc}
$N$ & $t$ & $\bar q_{t,N}$
& $10^3\varepsilon_{{\rm EDA},N}^{(7)}$
& $\alpha_{t,N}$ \\
\hline
5 & 10.6066 & 0.498574 & 4.441 & 0.503014 \\
6 & 19.8000 & 0.499640 & 1.764 & 0.501404 \\
7 & 25.0316 & 0.499909 & 0.879 & 0.500788 \\
8 & 52.8000 & 0.499977 & 0.392 & 0.500370 \\
9 & 69.5793 & 0.499994 & 0.308 & 0.500302 \\
\end{tabular}
\end{ruledtabular}
\end{table}

The estimated EDA discrepancy decreases monotonically with $N$, by a
factor of about 14 between $N=5$ and $N=9$.  This trend provides finite-size
evidence consistent with the $o(1)$ requirement in the EDA.

We extend the EDA test to $j=100$ for $N=6,7,8$, using 256 independently
drawn Hamiltonians at each size, with 16 histories per Hamiltonian and
2000 Hamiltonian-bootstrap replicates.  The evolution times are
$t=20.1000,29.2742,52.8000$ for $N=6,7,8$, respectively, all satisfying
$t\geq\sqrt d$.  We propagate a normalized representation of $\Delta_j$
to maintain numerical precision.  Figure~\ref{sm:fig:v2-long-eda} shows the
retention ratios.

For all three sizes, $R_j$ remains close to $1/2$ throughout the 100 updates.
For $N=8$, the maximum estimated discrepancy from $\bar q_t$ over this
window is $2.52\times10^{-4}$.
A separate $N=5$ test at $t=10.6066$ shows late-time deviations and weight
concentration on a few slowly decaying Hamiltonians, consistent with a
finite-size effect.

\begin{figure}[t]
\centering
\includegraphics[width=0.65\textwidth]{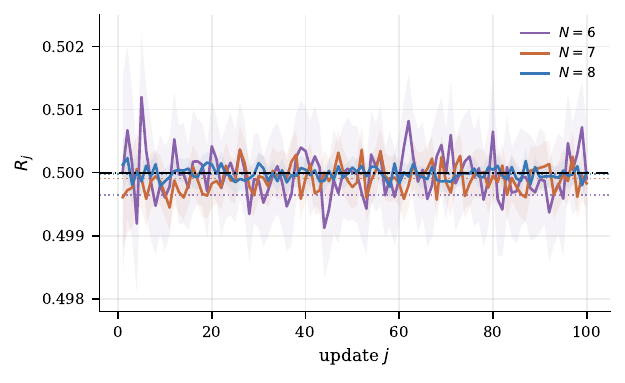}
\caption{\textbf{EDA tests over 100 updates.}
Retention ratios $R_j$ for $N=6,7,8$, using 256 Hamiltonians at each
size (solid curves).  Shading shows pointwise $95\%$ Hamiltonian-bootstrap
intervals.  Colored dotted lines are form-factor estimates of
$\bar q_t$, and the black dash-dotted line marks $1/2$.}
\label{sm:fig:v2-long-eda}
\end{figure}

Direct exponential fits of $\overline{\mathcal I}_\tau$ over
$1\leq\tau\leq12$ give decay factors close to $1/2$, confirming the
influence decay predicted by the EDA recursion.  Together, the finite-size
scaling and extended propagation tests provide quantitative numerical
validation of the EDA.

The linear-memory task $y_k=u_{k-\tau}$ in main-text Fig.~1(b) uses
$N=5,6,7$ at $t=10.6066,20.1000,29.2742$, respectively, with
$24,24,16$ Hamiltonians and two histories
per Hamiltonian.  Disjoint training, validation, and test blocks each contain
4000 inputs and are each preceded by a 500-step washout.
The readout uses the $N-1$ memory local-$Z$ features and the fixed raw-feature
penalty $\lambda_{\rm raw}=10^{-4}$.  We average the two histories at fixed
$H$, then resample Hamiltonians jointly across delays using 2000 bootstrap
replicates to obtain pointwise $95\%$ percentile intervals.
For each $N$, we fit the mean test performance to
$A_N\beta_{\rm perf,N}^{\tau-1}$ over $1\leq\tau\leq7$.
Here $\beta_{\rm perf,N}$ is the fitted factor by which the mean test
performance changes when the delay increases by one, and $A_N>0$ is the
fitted amplitude.  We average $P_\lambda$ over input histories
and Hamiltonian realizations before taking its logarithm, then fit a
straight line with a free intercept and slope.  Uncertainty is obtained
by resampling Hamiltonian realizations.  The largest fitted factor is
$\beta_{\rm perf,N}=0.447$, and all corresponding intervals lie below $1/2$.
This empirical performance decay factor differs from
$\alpha_t=\bar q_t+\varepsilon_{\rm EDA}$, which controls the influence bound.
The two factors need not coincide because the performance bound need not
be saturated.

\subsection{Quadratic no-go and linear-memory control}

The free-reservoir calculation propagates the full complex
$2^N\times2^N$ density matrix under the reset protocol.  It uses quadratic
Majorana Hamiltonians including pairing.  Their real antisymmetric coupling
matrices have independent Gaussian upper-triangular entries, and each
Hamiltonian is normalized to $2^{-N}\Tr(H_2^2)=1$.  The readout contains all Majorana
bilinears.

At $t=2$, $\tau=1$, and $\lambda_{\rm raw}=10^{-6}$, disjoint training and test
sequences of lengths $T=500,1000,2000,5000,10000$ are each preceded by a
separate 500-step washout.  The displayed $N=5$ data use 15 Hamiltonians and
two histories per Hamiltonian.  We first average the two histories at fixed
$H$; points and error bars in main-text Fig.~1(c) then show the mean and one
sample standard deviation across these 15 Hamiltonian-level values.
The estimated performance $P_\lambda$ approaches zero from below for the generic quadratic ensemble, as shown in
main-text Fig.~1(c).
To distinguish this obstruction from a loss of memory,
Fig.~\ref{sm:fig:v2-free-controls} applies the same quadratic-reservoir setting
to the linear-memory task $y_k=u_{k-\tau}$.  The small test errors show that past
inputs remain accessible to the readout.  The centered-product failure therefore
reflects the inability of quadratic dynamics and observables to generate the
required nonlinear cross-time dependence, rather than an absence of memory.

\begin{figure}[t]
\centering
\includegraphics[width=0.78\textwidth]{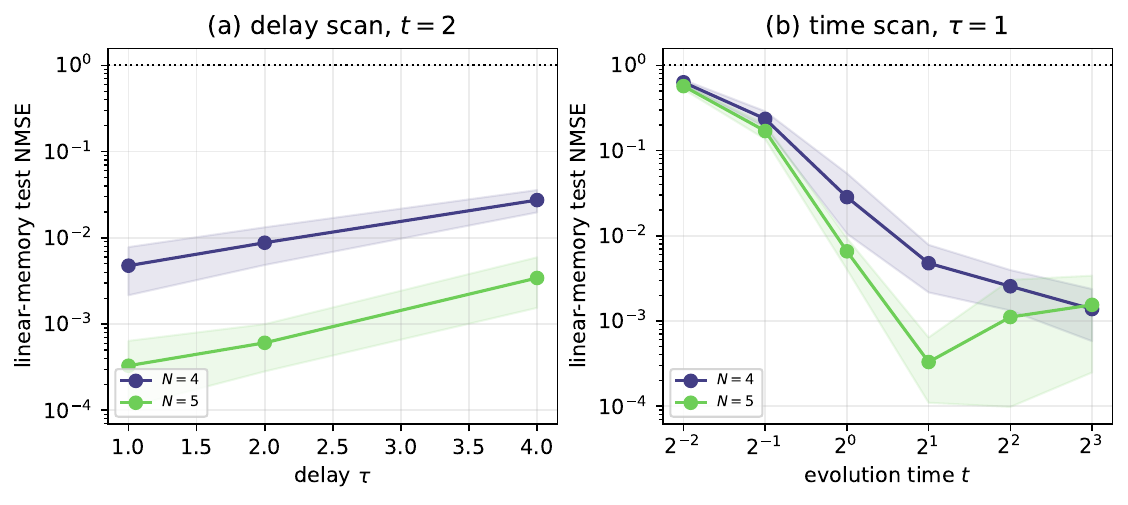}
\caption{\textbf{Linear-memory control for the free-reservoir obstruction.}
Past inputs remain predictable in the same quadratic-reservoir setting.
Both panels show $\mathrm{NMSE}_{\rm test}$ on a logarithmic scale.
Left: error versus delay at $t=2$.  Right: error versus evolution time at
$\tau=1$.  Both panels use generic quadratic Hamiltonians including pairing,
the complete Majorana-bilinear readout, $T=10000$, $N=4,5$, and
$\lambda_{\rm raw}=10^{-6}$.  Shading denotes $95\%$ bootstrap intervals
over Hamiltonian realizations.}
\label{sm:fig:v2-free-controls}
\end{figure}

\subsection{Population performance across the quadratic--quartic crossover}

The crossover calculation uses 30 independently drawn $N=6$
Sachdev--Ye--Kitaev (SYK) $\mathrm{SYK}_2$--$\mathrm{SYK}_4$ Hamiltonian
pairs.  For Majoranas $m_1,\ldots,m_{2N}$, the unnormalized Hamiltonians are
\begin{align}
 H_2^{(0)}&=\frac{\im}{2}\sum_{1\leq p<q\leq2N}
 A_{pq}m_pm_q,\nonumber\\
 H_4^{(0)}&=\sum_{1\leq p<q<r<s\leq2N}
 J_{pqrs}m_pm_qm_rm_s,
 \label{sm:eq:v2-numerical-syk-hamiltonians}
\end{align}
where all $A_{pq}$ and $J_{pqrs}$ are mutually independent real standard
Gaussian variables.  Equivalently, $A_{qp}=-A_{pq}$ defines a real
antisymmetric quadratic-coupling matrix.  We normalize the two sectors
separately as
\begin{equation}
 H_a:=
 \frac{H_a^{(0)}}{[2^{-N}\Tr\{(H_a^{(0)})^2\}]^{1/2}},
 \qquad a=2,4,
 \label{sm:eq:v2-numerical-syk-normalization}
\end{equation}
and use the interpolation $H_\theta$ defined in
Eq.~\eqref{sm:eq:v2-syk-interpolation}, with $g:=\tan\theta$.
The quadratic and pure-quartic Majorana sectors are Hilbert--Schmidt
orthogonal, so $2^{-N}\Tr(H_\theta^2)=1$ and $g$ gives the strength of the
quartic term relative to the quadratic term.  This is a generic Majorana
$\mathrm{SYK}_2+\mathrm{SYK}_4$ ensemble~\cite{garciagarcia2018chaotic,haque2019eigenstate}.  Particle-number conservation is not imposed: the
quadratic sector includes pairing, and the quartic sector includes
number-changing terms, although both sectors preserve fermion parity.  The
number-conserving complex-fermion interpolation
used in Ref.~\cite{kobayashi2026q_eoc} is a closely related,
symmetry-restricted variant.

Each Hamiltonian pair is reused across all couplings and lags.  The input is
uniform and i.i.d. under the amplitude encoding of
Eq.~\eqref{sm:eq:v2-amplitude-encoding}.  The readout is the complete set of 66
Majorana bilinears $\im m_pm_q$ with $p<q$, rescaled to the
Hilbert--Schmidt-orthonormal frame used in the analytic formulas.  We evaluate
Eq.~\eqref{sm:eq:v2-syk-interpolation} on
\begin{equation}
    g\in\{0,0.02,0.05,0.1,0.2,0.5,1,2,5,10,\infty\},
    \qquad t=2,\quad \tau\in\{4,5,6,7,8\},
    \quad \lambda_{\rm raw}=10^{-4},\quad
    \lambda=\lambda_{\rm raw}/64.
    \label{sm:eq:v2-crossover-design}
\end{equation}
Thus $g=0$ is the quadratic free-fermion endpoint, while $g=\infty$ is the
pure-quartic model in the chaotic regime, evaluated directly at $\theta=\pi/2$ rather than
approximated by a large finite $g$.
The horizontal axis in main-text Fig.~2 is linear in $\theta=\arctan g$,
with ticks labeled by $g$.
The inset in main-text Fig.~2 replots the $g=\infty$ case in the chaotic regime over the same
lag range on a logarithmic scale.  For every fixed Hamiltonian and coupling,
the stationary feature--target covariance is computed from
Lemma~\ref{sm:lem:v2-centered-product-covariance}, while the feature covariance
follows from the shared-input fixed point
Eq.~\eqref{sm:eq:v2-two-copy-fixed-point}.  Their combination in
Eq.~\eqref{sm:eq:v2-two-copy-population-gain} gives the stationary population
performance directly from the fixed-point moments.

Main-text Fig.~2 reports sample means; shading and error bars give pointwise
two-sided $95\%$ Student-$t$ intervals for the ensemble mean, using 29 degrees
of freedom and the standard error $\mathrm{SD}/\sqrt{30}$, where $\mathrm{SD}$
is the sample standard deviation across Hamiltonian realizations.  For every
displayed lag, $P_\lambda=0$ at the quadratic point and performance is enhanced at an
intermediate sampled coupling relative to the free endpoint and the sampled
chaotic regime.  In the chaotic regime at $g\gtrsim1$, it decreases
approximately exponentially over the lag window shown at the sampled
couplings $g=1,2,5,10,\infty$, consistent with rapid memory loss.
At $g=\infty$ in this regime, the task-dependent bound
$f_\lambda(\|\bm c_\tau\|^2/\sigma_y^2)$ from
Lemma~\ref{sm:lem:v2-sharp-envelope} is reasonably tight:
its sample mean is approximately $1.33$, $1.40$, and $1.45$ times
the mean performance at $\tau=4,5,6$, respectively.  Related
$\mathrm{SYK}_2$--$\mathrm{SYK}_4$ interpolations have been shown spectrally
to cross from free to chaotic dynamics and to exhibit QRC enhancement near
that boundary~\cite{kobayashi2026q_eoc}.  The sampled chaotic regime of the
$\mathrm{SYK}_2$--$\mathrm{SYK}_4$ interpolation supplies evidence for rapid
forgetting in a chaotic few-body ensemble, while the main-text GUE forgetting
theorem applies to GUE dynamics.  The two results provide complementary
evidence for rapid forgetting in the chaotic regime.  The crossover supports
intermediate-coupling enhancement for this
task family on the finite parameter grid, but does not establish a
task-independent or continuous-parameter optimum.  A sensitivity scan at
$\lambda_{\rm raw}\in\{10^{-5},10^{-4},10^{-3},10^{-2}\}$ retains a finite, nonzero sampled
maximum for every displayed lag, although the performance and the contrast between the free endpoint and
chaotic regime depend on the penalty.

\end{document}